\documentclass[a4paper,UKenglish,cleveref, autoref, thm-restate]{lipics-v2021}
\hideLIPIcs
\nolinenumbers
\usepackage{algorithm}
\usepackage{algpseudocodex}

\usepackage{epsfig}
\usepackage{epstopdf}
\usepackage{graphicx}
\DeclareGraphicsExtensions{.pdf,.png,.jpg,.eps}
\usepackage{latexsym}
\usepackage{amsfonts}
\usepackage{xcolor}
\usepackage[noadjust]{cite}
\usepackage{url}
\newcommand\ceil[1]{\left\lceil #1 \right\rceil}
\newcommand\floor[1]{\left\lfloor #1 \right\rfloor}

\Crefname{equation}{Equation}{Equations}
\crefname{equation}{Eq.}{Eqs.}

\def\problembox#1{%
    \vspace{2mm}%
    \noindent\fbox{%
    \begin{minipage}{.985\textwidth}%
        #1
    \end{minipage}%
    }%
    \vspace{2mm}%
}

\title{Indexing Range Maximum-Sum Segment Queries with Offsets} %
  \author{Seungbum Jo}{Chungnam National University, South Korea}{sbjo@cnu.ac.kr}{https://orcid.org/0000-0002-8644-3691}{}

 \author{Dominik K{\"{o}}ppl}{University of Yamanashi, Japan}{dkppl@yamanashi.ac.jp}{https://orcid.org/0000-0002-8721-4444}{This work was supported by JSPS KAKENHI Grant Number 25K21150.}

\authorrunning{S. Jo and D. K{\"{o}}ppl} %

\Copyright{Seungbum Jo and Dominik K{\"{o}}ppl} %



\relatedversion{} 

\funding{Seungbum Jo was supported by the National Research Foundation of Korea(NRF) grant funded by the Korea government(MSIT) (RS-202523963814)}%

\acknowledgements{%
A preliminary version containing some of these results has appeared in the proceedings of the 20th Scandinavian Symposium on Algorithm Theory (SWAT 2026)~\cite{DBLP:conf/swat/JoK26}.
The present article adds the time-space trade-off of \cref{sec:tradeoff}, 
the enumeration result of \cref{cor:range_omss_multiple}, the tight binary-alphabet bound of \cref{thm:bounded_alphabet}, and refined lower bounds. 
It also replaces the RLE reduction of the conference version by the weighted construction of \cref{sec:rle}. 
The earlier reduction used the unweighted run-value sequence and did not account for run lengths; 
\Cref{ex:rle} demonstrates why that reduction does not hold in general.
}

\begin{document}

\maketitle

\begin{abstract}
Given an array of $n$ real numbers, the maximum segment sum (MSS) problem is to find a contiguous subarray that has the largest sum. 
While the MSS problem can be solved optimally with Kadane's algorithm in $O(n)$ time, 
the study of its indexing version spawned new extensions such as
(a) retrieving the MSS after subtracting a query offset parameter for all array entries or
(b) retrieving the MSS for arbitrary query ranges.
We here study the combination of both problems (a) and (b), which requires retrieving the MSS for arbitrary query ranges after subtracting a query offset parameter for all array entries.
For that, we present an index whose query time is only slower than the best known for (a) by a factor of $O(\log n)$.
In detail, our index uses $O(n \log n)$ space, supports queries in $O(\log^2 n)$ time, and can be constructed in $O(n \log^3 n)$ time.
More generally, for every integer $d$ with $1\le d\le\ceil{\log_2 n}$, we give an $O(dn)$-space index with $O(dn^{1/d}\log n)$ query time; in particular, for every fixed $\varepsilon>0$, we obtain linear space and $O(n^\varepsilon\log n)$ query time.
As side results, we obtain the same time--space trade-off in terms of the number of runs of a run-length encoded input, deduce a solution for (a) that works in run-length compressed space and time, and prove a tight $\Theta(n^{2/3})$ bound on the number of non-compatible offsets for binary arrays.
Finally we give supportive lower bounds for our query problem, showing that there is only a polylogarithmic gap of improvement left.
\end{abstract}

\section{Introduction}\label{sec:intro}

Given an array $X$ of real numbers, the maximum-sum segment (MSS) problem
is to find an interval $[i,j]$ within $X$ that maximizes the sum of its elements, 
which we call an MSS of the sequence.
This problem found applications, among others, in image processing~\cite{fukuda01data}, 
pattern recognition~\cite{perumalla95parallel}, and biological sequence analysis~\cite{wang03segid}. 

The textbook algorithm for solving the MSS problem is Kadane's algorithm, cf.~\cite[4.4 Maximum Subarray Sum]{durr20competitive}.
Since the various applications often require adaptations to specific problem variations,
it is not surprising that the MSS problem has been studied in various settings, cf.~\cite{chao14maximumsum} for an encyclopedic survey of MSS and its variations.
Chen and Chao~\cite{chen2007range} considered an indexing problem, where it is allowed to preprocess $X$ and then answer queries on intervals of~$X$.
For that task, they devised a linear-time constructible data structure that supports constant-time queries of an MSS for any query interval. 
They also showed that this problem is linearly equivalent to the range-minimum query (RMQ) problem, meaning that
a data structure with construction time~$t_c(n)$ and query time~$t_q(n)$ solving the MSS problem can be transformed into a data structure solving RMQs with the time complexities $t_c(cn)$ and $t_q(cn)$ for some constant $c$, respectively.
Gawrychowski and Nicholson~\cite{gawrychowski2015encodings} gave a $\Theta(n)$-bit data structure that also achieves constant query time when only the answer interval, rather than its sum, is returned.
They also showed that $1.89113 n-\Theta(\log n)$ bits are necessary for supporting this query type.

A maximal local MSS is a local MSS that is not a segment of
any local MSS other than it, where a local MSS is a segment that has itself as its only MSS\@.
Ruzzo and Tompa~\cite{ruzzo99linear} gave an algorithm that finds all distinct maximal local MSSs in $O(n)$ time, and Sakai~\cite{sakai18maximal} designed a linear-time constructible data structure supporting
constant-time queries of the maximal local MSS of any given segment that contains any given position. 

Bengtsson and Chen~\cite{bengtsson07computing} showed that an arbitrarily given number of non-overlapping segments that maximize the sum of all their elements can be found in linear time. 
Yu et al.~\cite{yu21finding} considered the MSS problem where each element of the input sequence is
uncertain within a specific interval and proposed a linear-time algorithm for this problem.
The density of a segment is defined as the mean of all elements in the segment. 
Cheng et al.~\cite{cheng09optimal} considered the MSS problem with the condition that the density of the segment to be found is between given lower and upper bounds.
Bae~\cite{dissBaeSungEun} and Liu and Chao~\cite{liu08algorithms} considered returning $k$ non-overlapping segments that maximize the sum of all their elements.

Another variant is the maximum-density segment query~\cite{lin02efficient,goldwasser02fast,kim03linear} that asks to find a segment of length between a given lower and upper bound that maximizes the density. 
Chung and Lu~\cite{chung04optimal} showed that this query is solvable in linear time.

Another natural extension of the one-dimensional maximum segment sum problem is the maximum subarray problem, where the input is a two-dimensional array with length $n$ in each dimension, and the goal is to find a contiguous rectangular subarray that has the largest sum. 
This problem has been addressed by Bentley~\cite{bentley84persective}, who gave an $O(n^3)$-time algorithm for the problem.
Subsequently, the time complexity was improved to $O(n^3 \sqrt{\log \log n / \log n})$ by Takaoka~\cite{takaoka02efficient} and further to $O(n^3 \log \log n / \log n)$ by Tamaki and Tokuyama~\cite{tamaki00algorithms}. They showed that the problem can be reduced to $(\min,+)$-matrix multiplication. Using this reduction, Chan and Williams~\cite{DBLP:journals/talg/ChanW21} obtained the current fastest algorithm, running in $n^3/2^{\Omega((\log n)^{1/2})}$ time; their result also applies to higher-dimensional arrays.
In addition, Fukuda et al. \cite{fukuda01data} introduced the maximum convex sum problem, which is a variant of the maximum subarray problem where the retrieved subarray must be convex (i.e., the intersection of the subarray with any row or column is connected).
Here, Bae and Takaoka~\cite{bae17efficient} gave an $O(n)$-time parallel algorithm solving the maximum convex sum problem for $n^2$ processors.

\subparagraph*{Our Contribution}
Among the aforementioned variations, we focus on two indexing versions of the MSS problem, where we are allowed to preprocess the input array and then answer queries 
(a) on intervals of the input array, or (b) after subtracting a query offset parameter from all array entries.
In particular, we are interested in the combinations of both problems (a) and (b), where we want to answer queries on intervals of the input array after subtracting a query offset parameter from all array entries.
We present an indexing data structure solving this combined problem in \cref{thm:range_omss}.
The main tool is to answer this problem for the special case that the output interval has to be a prefix or a suffix of the input array, and to reuse known results for the problem (b) to solve the general case on $O(\log n)$ segments. 
We further generalize the canonical-range decomposition to a multiary hierarchy and obtain an $O(\tau n)$-space versus $O(\tau n^{1/\tau}\log n)$-query-time trade-off in \cref{thm:range_omss_tradeoff}; this yields linear space and $O(n^\varepsilon\log n)$ query time for every fixed $\varepsilon>0$.
The results are summarized in \Cref{tab:summary}.

Subsequently, we focus on run-length compressed input. 
For an input with $z$ runs, \cref{cor:range_rle_tradeoff} gives the same $O(\tau z)$-space versus $O(\tau z^{1/\tau}\log z)$-query-time trade-off.
We also obtain an $O(z)$-space OMSS structure in \cref{lem:sakai_rle}. 
Integer offsets can be handled with integer predecessor structures as described in \cref{remark:range_rle}.
For binary arrays, we show that the maximum number of pairwise non-compatible offsets is $\Theta(n^{2/3})$, 
which gives a tight bound on the size of the explicit OMSS breakpoint list (\cref{thm:bounded_alphabet}).

Finally, the predecessor reduction of \cref{thm:omss_predtradeoff} transfers predecessor lower bounds to OMSS and R-OMSS.
For arbitrary real keys, it gives an $\Omega(\log n)$ query lower bound in the comparison model.
For integer keys, it transfers the parameter-dependent cell-probe trade-off for predecessor search.  
In the encoding model, suffix R-OMSS and R-OMSS require $\Omega(n\log n)$ bits on permutations and $\Omega(n\log\sigma)$ bits over an alphabet of size $\sigma$ 
(\cref{thm:range_omss_lb,cor:range_omss_boundedalphabet}).  

\begin{table}
	\centering
		\begin{tabular}{c | c | c | c}
			\hline
			Query type & Space & Query time & Reference \\
			\hline
            MSS & $O(1)$ & $O(n)$ & Kadane's algorithm \\
            OMSS & $O(n)$ & $O(\log n)$  & Sakai~\cite{sakai2024data} \\
            R-OMSS without offset & $O(n)$ & $O(1)$ & Chen et al.~\cite{chen2007range} \\
            Suffix R-OMSS & $O(n)$ &$O(\log^2 n)$& \Cref{thm:range_suffix_omss} \\
            R-OMSS & $O(n \log n)$ & $O(\log^2 n)$ & \Cref{thm:range_omss} \\
            R-OMSS & $O(dn)$ & $O(dn^{1/d}\log n)$ & \Cref{thm:range_omss_tradeoff} \\
			\hline
	\end{tabular}
    \caption{Summary of best known complexities for the MSS query and its variants.}
	\label{tab:summary}
\end{table}

\section{Preliminaries and Query Definitions}\label{sec:prelim}

\begin{figure}[tp]
	\begin{center}
		\includegraphics[]{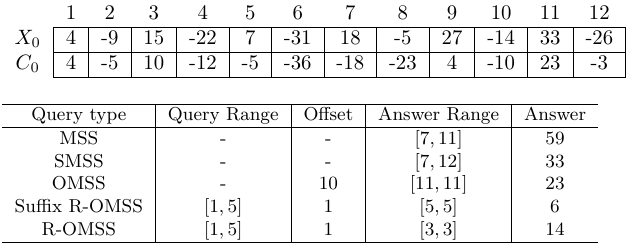}
	\end{center}
	\caption{Examples for the MSS query and its variants.
	}
    \label{fig:example_queries}
\end{figure}

Our computation model is the RAM model.
Let $\mathbb R$ denote the set of all real numbers.
For integers $i\le j$, we write $[i..j]:=\{i,i+1,\ldots,j\}$ and $|[i..j]|:=j-i+1$.
Closed intervals of real numbers retain the standard notation $[x,y]$.
We assume that all real numbers can be stored in a single word and that basic arithmetic operations on these numbers can be performed in constant time.
(We also assume that there are no errors in computational precision due to rounding or overflow.)

Let $X_0[1..n]$ be the input array of real numbers.
$X_0$ is subject to the following queries, for which we present examples in \cref{fig:example_queries}.
While the queries are initially stated to return the value of a specific subarray of $X_0$, 
it is beneficial to also return the interval of the subarray that gives the answer, which we call the \textit{answer range}.
Since we here mostly focus on data structures using $\Omega(n)$ space,
it suffices to focus on finding answer ranges --- we will shortly later see
that the actual value can be obtained from an answer range with an $O(n)$-space array of prefix sums.

For an offset $\alpha\in\mathbb R$ and a nonempty range $I := [\ell..r]$, define the \emph{score}
\(
    S_\alpha(I) := 
    S_\alpha(\ell,r)
    :=
    \sum_{i=\ell}^{r}(X_0[i]-\alpha).
\)
Whenever several candidate ranges attain the same maximum value, we prefer one of minimum
length.
In other words, an answer range is selected by maximizing
\begin{equation}\label{eq:shortest_answer}
    \bigl(
        S_\alpha(\ell,r),
        -(r-\ell+1)
    \bigr)
\end{equation}
lexicographically.
If several maximum ranges have the same minimum length, any one of them may be returned.
We call every such range a \emph{shortest answer range}.
The same tie-break can be applied to restricted versions such as the MSS problem by setting $\alpha=0$.
For prefix and suffix queries, the same rule is applied after restricting the candidate
ranges to prefixes or suffixes, respectively.
In particular, a shortest optimum prefix or suffix is unique because one endpoint is fixed.

The minimum-length condition has the useful consequence that a shortest answer range
contains neither a zero-sum proper prefix nor a zero-sum proper suffix:
deleting such a prefix or suffix would preserve the value and produce a shorter range.

We state our base problem as follows.

\problembox{%
  \textsc{Maximum-Segment-Sum (MSS)}\\
  \textbf{Input:} array of real numbers $X_0[1..n]$.\\
  \textbf{Output:} value and a shortest answer range $[\ell..r]$ achieving $\max_{[\ell..r] \subseteq [1..n]} \sum_{i=\ell}^r X_0[i]$.
}

We also study the special case that the answer range of the MSS problem is constrained to be a suffix of $X_0$.

\problembox{%
  \textsc{Suffix Maximum-Segment-Sum (SMSS)}\\
  \textbf{Input:} array of real numbers $X_0[1..n]$ \\
  \textbf{Output:} value and the starting position $\ell$ of a shortest suffix answer achieving $\max_{\ell \in [1..n]} \sum_{i=\ell}^n X_0[i]$.
}

Sakai~\cite{sakai2024data} introduced the following indexing problem that generalizes the MSS problem by allowing a query offset parameter $\alpha$ to be subtracted from all array entries before answering the query.

\problembox{%
  \textsc{Offset Maximum-Segment-Sum (OMSS)}\\
  \textbf{Construction Input:} array of real numbers $X_0[1..n]$ \\
  \textbf{Query Input:} real-value offset $\alpha \ge 0$.\\
  \textbf{Query Output:} value and a shortest answer range $[\ell..r]$ achieving $\max_{[\ell..r] \subseteq [1..n]} \sum_{i=\ell}^r (X_0[i] - \alpha)$.
}

Defining the array $X_\alpha[1..n]$ with $X_\alpha[i] := X_0[i] - \alpha$,
the OMSS query on $X_0$ with offset $\alpha$ is the MSS query on $X_\alpha$.
Hence, OMSS is a generalization of MSS, since the answer to an MSS query on $X_0$ is the same as the answer to an OMSS query on $X_{\alpha}$ with offset $\alpha = 0$. 
Sakai~\cite{sakai2024data} presented an approach answering an OMSS query as follows.

\begin{lemma}[\cite{sakai2024data}]\label{lemma:sakaiDS}
    Given an array $X_0[1..n]$ of $n$ real numbers, 
		we can compute a list $L$ of at most $n$ (offset, interval)-pairs 
    in $O(n \log^2 n)$ time using $O(n)$ space such that
    the solution to an OMSS query for a given offset $\alpha$ is the interval of the predecessor offset of $\alpha$ in $L$.
    By sorting $L$ descendingly by the offsets, 
    a binary search for the predecessor of $\alpha$ in $L$ can return an answer range of an OMSS query in $O(\log n)$ time.
\end{lemma}

Four observations are in order regarding \cref{lemma:sakaiDS}.

\begin{itemize}
    \item First, while the list of \cref{lemma:sakaiDS} does not explicitly return the MSS for a query offset $\alpha$,
    we can compute the sum $\sum_{i=a}^b (X_0[i] - \alpha)$ for an entry $(\alpha, [a..b])$ in $L$ in constant time
    if we have the array of prefix sums $C_0[0..n]$ of $X_0$, where $C_0[0] := 0$ and $C_0[i] := \sum_{k=1}^i X_0[k]$ for each $i \in [1..n]$.
    To see that, define the prefix sum of $X_\alpha$ for each $i\in[1..n]$ as
    \(
        C_\alpha[i]
        :=
        \sum_{k=1}^{i}X_\alpha[k]
        =
        C_0[i]-\alpha i.
    \)
    Then, for $1\le i\le j\le n$,
    \(
        S_\alpha(i,j)
        =
        C_0[j]-C_0[i-1]-\alpha(j-i+1),
    \)
    and hence the value of an answer range is obtained in constant time.
    
    \item Second, the upper bound $|L| \le n$ is tight since 
    $L$ can reach the size $n$ for the input 
    $X_0 = [1, 2, 3, \dots, n]$, where the OMSS query for $\alpha = k+\epsilon$ gives the range $[k+1..n]$ for every $\epsilon \in (0,1)$ and $k \in [0..n-1]$.
    
    \item Third, there can be several optimum ranges for one offset.
    The list of \cref{lemma:sakaiDS} stores one of minimum length.
    By \cref{eq:shortest_answer}, the stored range contains neither a zero-sum proper prefix nor a zero-sum proper suffix.
    However, Sakai's statement does not prescribe which optimum range is returned exactly at a breakpoint, which we here enforce by a modification that has no impact on the complexity of \cref{lemma:sakaiDS}.
\item 
Finally, Sakai allows returning the empty interval whenever no nonempty interval has positive score. 
In that case every entry of $X_\alpha$ is nonpositive, 
and a shortest nonempty optimum is the singleton $[i..i]$, where
$ i\in\arg\max_{j\in[1..n]}X_0[j]. $
We therefore replace the empty interval by this singleton.
\end{itemize}

For a range $I$, its \emph{score}
\( f_I(\alpha) := S_\alpha(I) \)
is an affine function.
Then the answer to the OMSS query with offset $\alpha$
is \( \max_I f_I(\alpha). \)
The graph of $\alpha \to \max_I f_I(\alpha)$, obtained by taking the pointwise maximum of all interval lines, is called the \emph{upper envelope}.
We further say that $f_I$ is \emph{active} at $\alpha$ if $f_I(\alpha) = \max_J f_J(\alpha)$.

Suppose that the lines of two scores meet at an offset $\beta$ in the upper envelope,
with $f_I$ active immediately to the left of $\beta$ and $f_J$ active immediately to its right.
Then $|J|<|I|$ (otherwise, the negative slope of $f_J$ would be at least that of $f_I$, and $f_I$ would be active immediately to the right of $\beta$ as well).
Also, $J$ is a shortest answer at $\beta$: if another optimum range $K$ at $\beta$
satisfied $|K|<|J|$, then $f_K$ would dominate $f_J$ immediately to the right of $\beta$.

Accordingly, we can associate every breakpoint with the envelope range active immediately to its right.
Thus, an entry $(\beta,I)$ represents a shortest answer from $\beta$ up to, but excluding, the next breakpoint.
If several ranges define the same affine function, 
they have the same value and length for every offset, and any one of them may be retained.
We call that the \emph{right-active convention}, and impose the right-active convention by a linear scan once we have constructed the list of \cref{lemma:sakaiDS}. 
Doing so does not change the time and space bounds of \cref{lemma:sakaiDS}.

We formalize the list $L$ of \cref{lemma:sakaiDS} to reuse the same concepts afterwards, which we term as a \emph{catalog}.
For a fixed interval $I$ and a family $\mathcal F(I)$ of candidate subranges of $I$, 
a \emph{catalog} is a sequence
\(
    \mathcal C_{\mathcal F}(I)
    =
    \bigl(
        (\beta_0,J_0),\ldots,(\beta_k,J_k)
    \bigr),
\)
where
\(
    0=\beta_0<\cdots<\beta_k,
\)
such that $J_j \in \mathcal F(I)$ is the selected shortest optimum for every offset
\(
    \alpha\in[\beta_j,\beta_{j+1}),
\)
with $\beta_{k+1} := +\infty$.
As aforementioned, at a breakpoint, we select the range immediately to its right.
Consequently, a predecessor search among the breakpoints returns the answer for any given offset $\alpha$.
A \emph{predecessor} of $\alpha$ is the largest offset in $L$ that is smaller than or equal to $\alpha$.
With the right-active convention, predecessor search returns a shortest answer also when $\alpha$ is exactly a breakpoint.

Now we present our main problem of interest, which combines the OMSS and MSS queries on subarrays of $X_0$. 

\problembox{%
  \textsc{Range Offset Maximum-Segment-Sum (R-OMSS)}\\
  \textbf{Construction Input:} array of real numbers $X_0[1..n]$ \\
  \textbf{Query Input:} real-value offset $\alpha \ge 0$ and a range $[a..b] \subseteq [1..n]$.\\
  \textbf{Query Output:} value and a shortest answer range $[\ell..r]$ achieving $\max_{[\ell..r] \subseteq [a..b]} \sum_{i=\ell}^r (X_0[i] - \alpha)$.
}
Using the notation of $S_\alpha(i,j)$, the answer to an R-OMSS query on $X_0$ with query range $[a..b]$ and offset $\alpha$ can be computed as $\max_{[\ell..r] \subseteq [a..b]} S_\alpha(\ell, r)$.
The R-OMSS is a generalization of the OMSS problem (with the query range $[1..n]$) and the range MSS problem introduced by Chen et al.~\cite{chen2007range}, 
which asks for the MSS of a query range $[a..b]$ without the offset parameter (i.e., $\alpha = 0$).
Like for MSS and its restriction SMSS, we can also define the suffix and prefix versions of R-OMSS queries, 
where the answer range is constrained to be a suffix or a prefix of the query range, respectively.

\problembox{%
  \textsc{Suffix Range Offset Maximum-Segment-Sum (Suffix R-OMSS)}\\
  \textbf{Construction Input:} array of real numbers $X_0[1..n]$ \\
  \textbf{Query Input:} real-value offset $\alpha \ge 0$ and a range $[a..b] \subseteq [1..n]$.\\
  \textbf{Query Output:} value and the starting position $\ell$ of a shortest suffix answer achieving $\max_{\ell \in [a..b]} \sum_{i=\ell}^b (X_0[i] - \alpha)$.
}

A symmetric definition can be given for \emph{prefix R-OMSS} queries, where the answer range is constrained to be a prefix of the query range.

\section{Data structure for R-OMSS Queries}\label{sec:range_omss}

Our main goal of this paper is to devise data structures for answering R-OMSS queries on $X_0$, 
where a non-negative offset $\alpha \ge 0$ and a query range $[a..b]$ are given at query time. 
There are two naive approaches prioritizing either (1) space or (2) time.
\begin{enumerate}[(1)]
    \item We can apply Kadane's algorithm on $X_{\alpha}[a..b]$. 
    Since any value in $X_{\alpha}$ can be computed in $O(1)$ time from $X_0$, 
    Kadane's algorithm answers R-OMSS in $O(n)$ time using $O(1)$ space additional to $X_0$.
\item Alternatively, we can maintain the data structure of \cref{lemma:sakaiDS} for all $\Theta(n^2)$ possible query ranges. This results in an $O(n^3)$-space data structure with $O(\log n)$ query time.
\end{enumerate}

In what follows, we present two non-trivial data structures for R-OMSS queries whose complexities lie between those of solutions (1) and (2).
We start with an $O(n)$-space  data structure answering suffix R-OMSS queries on $X_0$ in $O(\log^2 n)$ time.
By symmetry, the same data structure results can be obtained for prefix R-OMSS queries.
Next, in \cref{sec:general}, we consider a data structure for (general) R-OMSS queries.
By combining the data structure of \cref{lemma:sakaiDS} with the suffix R-OMSS data structures, we obtain an $O(n \log n)$-space data structure that can answer an R-OMSS query in $O(\log^2 n)$ time.

Finally, note that all our data structure results can also handle the case $\alpha < 0$ within the same asymptotic space bounds, by maintaining the negated input array and answering the corresponding minimum segment-sum query with offset $-\alpha$. All the data structures for the minimum segment sum variants are analogous to those for the maximum segment sum variants.

\subsection{Data Structure for Suffix R-OMSS Queries}\label{sec:suffix}

\begin{figure}[tp]
	\begin{center}
            \includegraphics[scale=0.8]{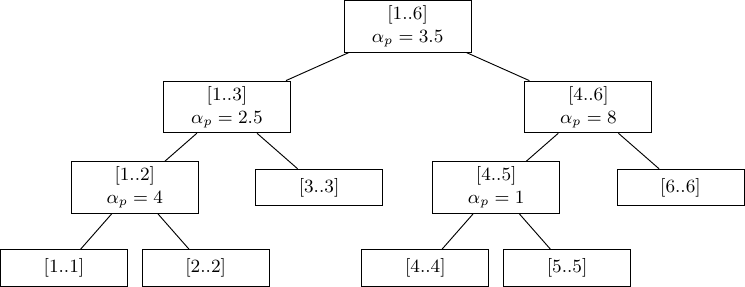}
	\end{center}
    \caption{SBST on $X_0 = [4, 1, 6, 1, 8, 0]$.
    Each node is labeled with its corresponding canonical range.
    Each internal node is additionally labeled with its threshold offset.
}
    \label{fig:sbst}
\end{figure}

We start with devising a data structure for range suffix OMSS queries.
The main idea is to maintain the canonical ranges (defined below) of $X_0$ using a \textit{suffix binary search tree (SBST)}. 
Each internal node in the tree stores a \textit{threshold offset}.
By comparing $\alpha$ with the offsets stored along a root-to-leaf path in the tree, as in a standard binary search tree, we can track the endpoint of the range corresponding to the query answer when the query range is canonical.
To map this logic to arbitrary query intervals, we define a \emph{canonical cover} of a range~$[a..b]$ as a set of canonical ranges that are pairwise disjoint and whose union is $[a..b]$.
Finally, we show that we can partition the query range~$[a..b]$ into a canonical cover of size $O(\log n)$, and answer the query using a dynamic programming approach.

The canonical ranges of $X_0$ are the elements of a recursive binary partition of $[1..n]$ into $O(n)$ ranges, 
which can be represented by a binary tree. 
The \emph{canonical ranges} of $X_0$ are defined as follows: 
\begin{enumerate}[(1)]
    \item $[1..n]$ is a canonical range, and 
    \item for every canonical range $[g..h]$ with $g < h$, 
    both $[g..g+\floor{(h-g)/2}]$  
    and $[g+\floor{(h-g)/2}+1..h]$ are also canonical ranges. 
\end{enumerate}
Under this definition, there are $O(n)$ canonical ranges in total. 
We then define the SBST on $X_0[g..h]$ as follows:

\begin{itemize}
    \item The root of the SBST on $X_0[g..h]$ is the canonical range $[g..h]$.
    \item If $g < h$, then the SBSTs on $X_0[g..g+\floor{(h-g)/2}]$ and $X_0[g+\floor{(h-g)/2}+1..h]$ are the left and right child subtrees of the root, respectively. 
        Otherwise, for $g=h$, the SBST root itself is a leaf with canonical range $[g..g]$.
\end{itemize}
See \cref{fig:sbst} for an example.
(The reader familiar with segment trees may notice that the SBST is a specialization of a segment tree indexing all integers in $[1..n]$.)

Let $T$ denote the SBST on $X_0[1..n]$.
From the definition, each canonical range of $X_0$ is a node in $T$, and vice versa.
The leaves of $T$ are all canonical ranges of size $1$, i.e., $[i..i] = \{i\}$ for each $i \in [1..n]$.
By construction, $T$ is a binary tree with $O(n)$ nodes and height $O(\log n)$.
Finally, we augment each node $p$ of $T$ with a \textit{threshold offset} $\alpha_p$. 
Before defining the threshold offsets, we introduce the following key property for answering SMSS queries.

\begin{proposition}\label{prop:sbst}
For any non-negative offsets $\alpha<\beta$, the shortest answer range of the SMSS query
on $X_\beta$ is a subrange of the shortest answer range of the SMSS query on $X_\alpha$.
\end{proposition}
\begin{proof}
Let $[x..n]$ and $[x'..n]$ be the shortest answer ranges for offsets $\alpha$ and
$\beta$, respectively.
Suppose that $x'<x$.
The range $[x..n]$ is shorter than $[x'..n]$.
Since $[x'..n]$ is the shortest answer at offset $\beta$, while $[x..n]$ is shorter
than $[x'..n]$, we must have
\(
    S_\beta(x',n)>S_\beta(x,n)
\)
--- in case of equality, we would already favor the shorter range $[x..n]$ by \cref{eq:shortest_answer}.
It follows that $S_\beta(x',x-1)>0$, and hence
\[
    S_\alpha(x',x-1)
    =
    S_\beta(x',x-1)+(\beta-\alpha)(x-x')
    >0.
\]
Consequently,
$S_\alpha(x',n)>S_\alpha(x,n)$, contradicting the optimality of $[x..n]$.
Therefore $x'\ge x$.
\end{proof}

Let $[x_p..y_p]$ be the range corresponding to an internal node $p$, and set
\(
    z_p:=x_p+\floor{(y_p-x_p)/2}.
\)
For an offset $\alpha\ge0$, let $[x_\alpha..y_p]$ be the shortest SMSS answer on
$X_\alpha[x_p..y_p]$.
We define the threshold offset of $p$ by
\(
    \alpha_p
    :=
    \inf\{\alpha\ge0\mid x_\alpha\in[z_p+1..y_p]\}.
\)
The set in this definition is nonempty, since for sufficiently large $\alpha$ the shortest SMSS answer is the singleton $[y_p..y_p]$.
By \cref{prop:sbst}, $x_\alpha$ is non-decreasing as a function of $\alpha$.
At the breakpoint at which a suffix from the left child and a suffix from the right child have equal value, 
the latter suffix is shorter and is therefore the unique shortest SMSS answer.
Consequently,
\begin{equation}\label{eq:sbst_threshold_rule}
    x_\alpha\in[x_p..z_p]
    \Longleftrightarrow
    \alpha<\alpha_p,
\end{equation}
and $x_\alpha\in[z_p+1..y_p]$ if and only if $\alpha\ge\alpha_p$.
Leaves store no threshold.

The data structure for suffix R-OMSS queries is composed of 
\begin{enumerate}[(a)]
    \item the array $C_0$, and  \label{it:c0}
    \item the SBST $T$ on $X_0$. \label{it:sbst}
\end{enumerate}
Each node of the SBST stores its threshold offset and its canonical range.
(\ref{it:c0}) and (\ref{it:sbst}) together use $O(n)$ space in total.

We now describe how to answer a query using this data structure.
We first consider the case where the query range is a canonical range $[x_p..y_p]$ corresponding to a node $p$ in $T$.
In this case, it suffices to find the position $x$ such that $[x..y_p]$ is the answer range.
To find $x$, we start from the node $p$ and compare $\alpha$ with $\alpha_p$.
As in a standard binary search tree, we recursively perform the same comparison:
we proceed to the left child if $\alpha<\alpha_p$, and to the right child otherwise,
until we reach a leaf node corresponding to the canonical range $[x'..x']$.
Then $x' = x$ from the following proposition.

\begin{proposition}\label{prop:sbst2}
Let $p$ be an internal node with range $[x_p..y_p]$ and midpoint $z_p$.
If $\alpha<\alpha_p$, then the shortest SMSS answers on
$X_\alpha[x_p..y_p]$ and $X_\alpha[x_p..z_p]$ have the same starting position.
If $\alpha\ge\alpha_p$, then the shortest SMSS answers on
$X_\alpha[x_p..y_p]$ and $X_\alpha[z_p+1..y_p]$ have the same starting position.
\end{proposition}
\begin{proof}
We prove the first case; the second one is symmetric.
By \cref{eq:sbst_threshold_rule}, the shortest SMSS answer $[x..y_p]$ on
$X_\alpha[x_p..y_p]$ starts in the left child.
Let $[x'..z_p]$ be the shortest SMSS answer on $X_\alpha[x_p..z_p]$.

Suppose first that $x<x'$.
If $S_\alpha(x,y_p)=S_\alpha(x',y_p)$, then the shorter range $[x'..y_p]$ would be
the shortest answer for the parent interval.
Therefore $S_\alpha(x,y_p)>S_\alpha(x',y_p)$, which implies
$S_\alpha(x,x'-1)>0$.
Hence
\[
    S_\alpha(x,z_p)
    =
    S_\alpha(x',z_p)+S_\alpha(x,x'-1)
    >
    S_\alpha(x',z_p),
\]
contradicting the optimality of $[x'..z_p]$.

Suppose now that $x>x'$.
If $S_\alpha(x',z_p)=S_\alpha(x,z_p)$, then the shorter range $[x..z_p]$ would be
the shortest answer for the left child.
Thus $S_\alpha(x',z_p)>S_\alpha(x,z_p)$, and adding the common suffix
$[z_p+1..y_p]$ yields
$S_\alpha(x',y_p)>S_\alpha(x,y_p)$, contradicting the optimality of the parent suffix.
Therefore $x=x'$.
\end{proof}

Since the height of $T$ is $O(\log n)$, we can answer a suffix R-OMSS query on $X_0$ in $O(\log n)$ time when the query range is canonical.

Finally, we consider how to answer a suffix R-OMSS query on $X_{0}$ 
for a general query range $[a..b]$ using the data structure.
We first partition $[a..b]$ into $c = O(\log n)$ distinct canonical ranges 
$[a_1..b_1], \dots, [a_c..b_c]$, ordered from left to right (i.e., $a_1 = a$, $b_c = b$, and $b_k+1 = a_{k+1}$ for all $k\in[1..c-1]$).
This partition can be computed by starting from the root of $T$ and traversing the tree downwards in two directions for finding the leaves of the canonical ranges $[a..a]$ and $[b..b]$,
respectively. 
We then collect $O(\log n)$ canonical ranges corresponding to children on the two root-to-leaf paths.
At each level of $T$, there are at most two nodes whose canonical ranges $[x..y]$ satisfy $[x..y] \cap [a..b] \neq \emptyset$ but $[x..y] \not\subseteq [a..b]$; 
only these nodes require recursion to the next level, while any other intersecting node is fully contained in $[a..b]$ and can be reported as a canonical range immediately.
Therefore, we can compute these canonical ranges in $O(\log n)$ time. 

If $c = 1$ (i.e., $[a..b]$ is a canonical range), 
we can answer the R-OMSS query in $O(\log n)$ time using $T$.
Otherwise, we apply a dynamic programming approach as follows.
Define an array $A[1..c]$ of size $c$, where for each $k \in [1..c]$, $A[k]$ stores the answer range of the SMSS query on $X_{\alpha}[a_1..b_k]$.
Then $A[c]$ stores the answer of the suffix R-OMSS query.

To compute $A[c]$, we first handle the base case by computing $A[1]$ in $O(\log n)$ time using $T$.
Next, for $k > 1$, we compute $A[k]$ from $A[k-1]$ as follows.
To compute $A[k]$, it suffices to determine the answer range $[x..b_k]$ of the SMSS query on $X_{\alpha}[a_1..b_k]$. 
Now we introduce the following lemma, which is used to compute $A[k]$.

\begin{figure}[t]
    \centering
    \includegraphics[width=0.8\textwidth]{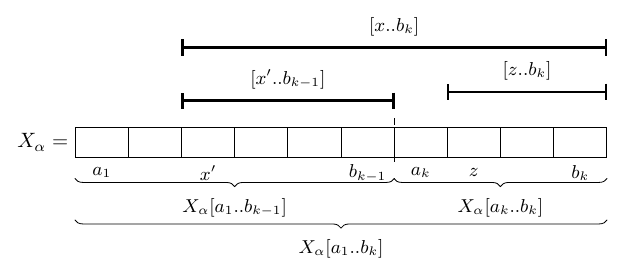}
    \caption{A sketch of the setting in \cref{lem:computak} with the setting $x = x'$.}
    \label{fig:computak}
\end{figure}

\begin{lemma}\label{lem:computak}
Let $[x..b_k]$ be the shortest SMSS answer on $X_\alpha[a_1..b_k]$, let
$[x'..b_{k-1}]$ be the shortest SMSS answer on $X_\alpha[a_1..b_{k-1}]$, and let
$[z..b_k]$ be the shortest SMSS answer on $X_\alpha[a_k..b_k]$.
Then $x\in\{x',z\}$.
\end{lemma}
\begin{proof}
    See \cref{fig:computak} for a sketch of the setting.
If $x\le b_{k-1}$, every suffix of $X_\alpha[a_1..b_k]$ starting at or before
$b_{k-1}$ consists of a suffix of $X_\alpha[a_1..b_{k-1}]$ followed by the fixed block
$X_\alpha[a_k..b_k]$.
Adding the same block preserves both the value order and the length order among ties.
Hence, the shortest SMSS answer among these candidates starts at $x'$.

If $x\ge a_k$, the candidate lies completely inside $X_\alpha[a_k..b_k]$, and by
definition its shortest SMSS answer starts at $z$.
These are the only two cases because $a_k=b_{k-1}+1$.
\end{proof}

The answer ranges of the SMSS queries on $X_{\alpha}[a_1..b_{k-1}]$ and $X_{\alpha}[a_k..b_k]$ can be obtained from $A[k-1]$ and from $T$, respectively, in $O(\log n)$ time.
Thus, by \cref{lem:computak}, given $A[k-1]$, the value $A[k]$ can be computed in $O(\log n)$ time, which implies that $A[c]$ can be computed in $O(\log^2 n)$ time.

\begin{figure}[tp]
	\begin{center}
		\includegraphics[width=\linewidth]{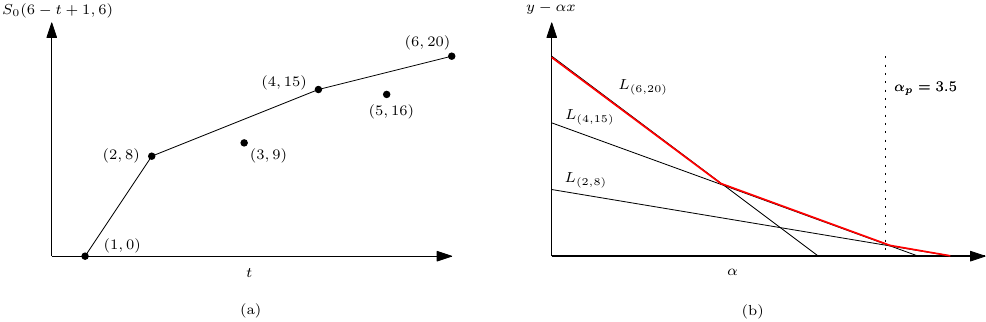}
	\end{center}
    \caption{Example illustrating the computation of the threshold offset at the root of $T$ in \cref{fig:sbst}. (a) The point set $P$ with its convex hull. (b) The dual lines corresponding to the convex hull of $P$. The red segments represent the upper envelope of these lines.}
    \label{fig:construct_t}
\end{figure}

\subparagraph*{Construction time.}
As $C_0$ can be constructed in $O(n)$ time, it suffices to analyze the construction time of $T$. To construct $T$, for each node $p$ whose corresponding range is $[x_p..y_p]$ with $\ell_p = y_p - x_p + 1$, we compute its threshold offset $\alpha_p$ as follows.
Consider the point set
$P = \{(t, S_{0}(y_p - t + 1, y_p)) \mid 1 \le t \le \ell_p\}$,
that is, each point corresponds to a suffix of $[x_p..y_p]$, where the $x$-coordinate $t$ is the suffix length and the $y$-coordinate is its sum.
We keep the vertices of the upper convex hull of $P$, since only hull vertices can maximize $y-\alpha x$ for some $\alpha\ge0$.
The convex hull can be constructed in $O(\ell_p)$ time, since the points in $P$ are already sorted by their first coordinates~\cite{graham1972efficient}. 

For any suffix represented by a point $(x, y) \in P$, we have
$S_{\alpha}(y_p - x + 1, y_p) = y - \alpha x$.
Thus, given an offset $\alpha$, answering the SMSS query on $X_0[x_p..y_p]$ reduces to finding the point $(x, y) \in P$ that maximizes $y - \alpha x$.
By point-line duality, each point $(x, y)$ corresponds to the line
$L_{(x, y)}(w) = y - w x$,
and $\max_{(x, y) \in P} (y - \alpha x)$ equals the value of the upper envelope at $w = \alpha$.
The upper envelope is piecewise linear and changes only at intersection points of consecutive envelope lines.

We first determine the shortest optimal suffix at offset $\alpha=0$.
It is represented by a point of maximum second coordinate, choosing the one of minimum first coordinate in case of a tie; this point can be identified during the hull construction.
If this suffix starts in the right half of $[x_p..y_p]$, then the definition of the threshold gives $\alpha_p=0$.
Otherwise, the shortest optimal suffix at offset $0$ starts in the left half.
Since it eventually becomes the singleton $[y_p..y_p]$, \cref{prop:sbst} implies that $\alpha_p>0$ is the first positive envelope breakpoint at which the right-active shortest optimum starts in the right half.
At that breakpoint, the right-half suffix is selected because it is shorter.
Since the envelope breakpoints are sorted by slope, this breakpoint can be found by binary search over the hull in $O(\log\ell_p)$ time.
See \cref{fig:construct_t} for an example.

Thus, the nodes at each level of $T$ can be constructed in $O(n)$ time, implying that the entire tree $T$ can be constructed in $O(n \log n)$ time overall.

We summarize the results of this section in the following theorem.

\begin{theorem}\label{thm:range_suffix_omss}
Given an array $X_0$ of size $n$, there exists an $O(n)$-space data structure that can answer a prefix/suffix R-OMSS query on $X_0$ in $O(\log^2 n)$ time for any non-negative offset $\alpha \ge 0$.
When the query range is a canonical range of $X_0$, the data structure can answer a prefix/suffix R-OMSS query in $O(\log n)$ time.
We can construct this data structure in $O(n \log n)$ time.
\end{theorem}

\subsection{Data Structure for General R-OMSS Queries}\label{sec:general}

For a nonempty interval $I=[g..h]$ and an offset $\alpha$, define the \emph{summary} of $I$ by
\(
    \mathcal S_\alpha(I)
    :=
    \bigl(
        S_\alpha(g,h),
        \mathsf{pre}_\alpha(I),
        \mathsf{suf}_\alpha(I),
        \mathsf{best}_\alpha(I)
    \bigr),
\)
where $\mathsf{pre}_\alpha(I)$ and $\mathsf{suf}_\alpha(I)$ are the maximum values of a prefix and a suffix of $X_\alpha[g..h]$, respectively, and
$\mathsf{best}_\alpha(I)$ is the maximum value of a nonempty segment of $X_\alpha[g..h]$.
Together with each of these three values, we store a shortest answer range attaining it according to \cref{eq:shortest_answer}.

Given two adjacent nonempty intervals $I=[g..h]$ and $J=[h+1..k]$,
we can compute the summary of $I \cup J$ in constant time by the following equations.
\begin{align}
    S_\alpha(g,k)
        &= S_\alpha(g,h)+S_\alpha(h+1,k), \label{eq:summary_total}\\
    \mathsf{pre}_\alpha(I\cup J)
        &= \max\bigl(
            \mathsf{pre}_\alpha(I),
            S_\alpha(g,h)+\mathsf{pre}_\alpha(J)
        \bigr), \label{eq:summary_prefix}\\
    \mathsf{suf}_\alpha(I\cup J)
        &= \max\bigl(
            \mathsf{suf}_\alpha(J),
            \mathsf{suf}_\alpha(I)+S_\alpha(h+1,k)
        \bigr), \label{eq:summary_suffix}\\
    \mathsf{best}_\alpha(I\cup J)
        &= \max\bigl(
            \mathsf{best}_\alpha(I),
            \mathsf{best}_\alpha(J),
            \mathsf{suf}_\alpha(I)+\mathsf{pre}_\alpha(J)
        \bigr). \label{eq:summary_best}
\end{align}
Each maximum in \cref{eq:summary_prefix,eq:summary_suffix,eq:summary_best} is resolved by
the lexicographic order defined by \cref{eq:shortest_answer}.
For example, the range corresponding to
$\mathsf{suf}_\alpha(I)+\mathsf{pre}_\alpha(J)$ is obtained by concatenating the suffix
range stored with $\mathsf{suf}_\alpha(I)$ and the prefix range stored with
$\mathsf{pre}_\alpha(J)$.
We prove correctness of \cref{eq:summary_total,eq:summary_prefix,eq:summary_suffix,eq:summary_best} together with the following claim.

\begin{lemma}\label{lem:mss_summary_merge}
Given $\mathcal S_\alpha(I)$ and $\mathcal S_\alpha(J)$ for two adjacent nonempty intervals $I$ and $J$, we can compute $\mathcal S_\alpha(I\cup J)$ in constant time.
\end{lemma}
\begin{proof}
Every prefix of $I\cup J$ is either contained in $I$, or consists of all of $I$
followed by a prefix of $J$.
Within the second class, using the shortest prefix-MSS answer on $J$ first maximizes the value and then
minimizes the resulting length.
Comparing minimum-length representatives of the two classes according to
\cref{eq:shortest_answer} proves \cref{eq:summary_prefix}.
The argument for \cref{eq:summary_suffix} is symmetric.

Every segment of $I\cup J$ is either contained in $I$, contained in $J$, or crosses their
common boundary.
Among the crossing segments, the value is maximized by combining a maximum suffix of $I$
with a maximum prefix of $J$.
Choosing minimum-length representatives also minimizes the length of the concatenation.
If several candidates have the same value and length, any one of them may be retained.
Comparing this crossing candidate with the two minimum-length internal candidates proves
\cref{eq:summary_best}.
Finally, \cref{eq:summary_total} follows from additivity of the score $S_\alpha$.
\end{proof}

We now construct the data structure for general R-OMSS queries.
It consists of the following three parts:
\begin{enumerate}[(1)]
    \item the prefix-sum array $C_0$,
    \item the data structures of \cref{thm:range_suffix_omss}, which answer prefix and suffix R-OMSS queries in $O(\log n)$ time when the query range is canonical, and
    \item for every canonical range $[g..h]$, the $O(h-g+1)$-space OMSS data structure of \cref{lemma:sakaiDS} on $X_0[g..h]$.
\end{enumerate}
The total length of all canonical ranges is $O(n\log n)$.
Consequently, (3) uses $O(n\log n)$ space, while (1) and (2) use $O(n)$ space.

Given a query range $[a..b]$, we partition it in $O(\log n)$ time into
$c=O(\log n)$ pairwise disjoint canonical ranges
\(
    I_1=[a_1..b_1],\ldots,I_c=[a_c..b_c]
\)
ordered from left to right.
For every $k\in[1..c]$, we obtain the summary $\mathcal S_\alpha(I_k)$ as follows.

\begin{itemize}
    \item We compute its total value
    \(
        S_\alpha(a_k,b_k)
        =
        C_0[b_k]-C_0[a_k-1]-\alpha(b_k-a_k+1)
    \)
    in constant time.
    \item We obtain its prefix and suffix fields $\mathsf{pre}$ and $\mathsf{suf}$ in $O(\log n)$ time from the data structures of \cref{thm:range_suffix_omss}, since $I_k$ is canonical.
    \item We obtain its best field $\mathsf{best}$ in $O(\log n)$ time from the OMSS structure stored for $I_k$.
\end{itemize}

We merge the $c$ summaries from left to right using \cref{lem:mss_summary_merge}.
After processing $I_1,\ldots,I_k$, the maintained summary is
\[
    \mathcal S_\alpha(I_1\cup\cdots\cup I_k).
\]

\subparagraph*{Complexity Analysis}
After all summaries have been merged, its $\mathsf{best}$ field is an answer to the R-OMSS query on $[a..b]$.
The $O(\log n)$ summaries each require $O(\log n)$ query time, while all merge operations together require only $O(\log n)$ time.
The total query time is therefore $O(\log^2 n)$.

The construction time of the data structure in \cref{lemma:sakaiDS} on an interval of length $\ell$ is $O(\ell\log^2\ell)$.
Since the total length of all canonical ranges is $O(n\log n)$, all structures in (3) can be constructed in $O(n\log^3n)$ time~\cite{sakai2024data}.
The remaining components require $O(n\log n)$ construction time.

We summarize the result in the following theorem.

\begin{theorem}\label{thm:range_omss}
Given an array $X_0$ of size $n$, there exists an $O(n\log n)$-space data structure that can answer an R-OMSS query on $X_0$ in $O(\log^2n)$ time for any non-negative offset $\alpha \ge 0$.
We can construct this data structure in $O(n\log^3n)$ time.
\end{theorem}

\subparagraph*{Returning Multiple Solutions.}
Although a query returns a shortest answer range, we can also report all optimum ranges that contain neither a zero-sum proper prefix nor a zero-sum proper suffix.  
We call such an optimum range \emph{irreducible}.
While a shortest answer is irreducible, an irreducible optimum range need not be a globally shortest answer.

\begin{proposition}\label{prop:multiplesolution}
    Any two distinct irreducible answer ranges of an R-OMSS query are disjoint.
\end{proposition}
\begin{proof}
Let $I=[\ell_1..r_1]$ and $J=[\ell_2..r_2]$ be two distinct overlapping answer ranges.

If $\ell_1=\ell_2$, then $r_1\ne r_2$ because $I\ne J$. 
After swapping $I$ and $J$ if necessary, assume $r_1<r_2$.
Since $I$ and $J$ have the same value,
\(
S_\alpha(r_1+1, r_2) = S_\alpha(J)-S_\alpha(I)=0.
\)
Thus, $[r_1+1..r_2]$ is a zero-sum proper suffix of $J$, 
contradicting irreducibility. 
This contradiction shows that $\ell_1\ne\ell_2$. 

After swapping $I$ and $J$ if necessary, assume $\ell_1<\ell_2$.
Let the common maximum score be $S_\alpha(J) = S_\alpha(I)= M$, and consider the following two cases.

\begin{itemize}
    \item First suppose that $r_1<r_2$.
    Write
    \(A:=[\ell_1..\ell_2-1],\)
    \(B:=[\ell_2..r_1],\) and
    \(C:=[r_1+1..r_2], \)
    and let $a := S_\alpha(A)$, $b := S_\alpha(B)$, and $c := S_\alpha(C)$
    be their scores.
    Since $a+b=M=b+c$, we have $a=c$.
    If $a>0$, then the union $A\cup B\cup C$ has value $M+a>M$.
    If $a<0$, then the intersection $B$ has value $M-a>M$.
    Both cases contradict maximality.
    Hence, $a=c=0$, so $A$ is a zero-sum proper prefix of $I$ and $C$ is a zero-sum proper suffix of $J$.
    
    \item Now suppose that $r_2\le r_1$; this includes both nested ranges and the case of equal right endpoints.
    Write
    \( A := [\ell_1..\ell_2-1], \)
    \( B := [\ell_2..r_2], \) and
    \( C := [r_2+1..r_1], \)
    where $C$ may be empty.
    Set $c:=0$ if $C$ is empty, and $c:=S_\alpha(C)$ otherwise.
    Further, set $a := S_\alpha(A)$ and $b := S_\alpha(B)$.
    Since $b=M$ and $a+b+c=M$, we have $a+c=0$.
    If $a>0$, then $A\cup B$ has value $M+a>M$.
    If $a<0$, then $B\cup C$ has value $M-a>M$.
    Thus, $a=0$, and $A$ is a zero-sum proper prefix of $I$.
    This contradicts the assumption in all cases.
\end{itemize}
\end{proof}

Let $M$ be the value of an initial R-OMSS query on $[a..b]$.
After reporting an answer range $[\ell..r]$, recursively query the two remaining ranges
$[a..\ell-1]$ and $[r+1..b]$, omitting empty ranges.
A recursive branch is stopped as soon as its answer value is smaller than $M$.
By \cref{prop:multiplesolution}, every other answer range is disjoint from $[\ell..r]$ and is therefore fully contained in one of the two remaining ranges.
Conversely, every range reported by the recursion has value $M$ and is an answer to the original query.
The recursion performs only a constant number of queries per reported range.

\begin{corollary}\label{cor:range_omss_multiple}
Given an array $X_0$ of size $n$,
we can enumerate all distinct irreducible answer ranges of an R-OMSS query on $X_0$ with the data structure of \cref{thm:range_omss}.
The delay is $O(\log^2n)$ per reported range.
We can report up to $g$ distinct answer ranges in \( O(g\log^2n) \) time.
\end{corollary}
\begin{proof}
    We eagerly query both residual ranges and recurse only on those whose value is $M$.
\end{proof}

\section{A Time--Space Trade-Off}\label{sec:tradeoff}

The data structure of \cref{thm:range_omss} stores an OMSS data structure for every canonical range of the binary partition of $[1..n]$.
Since the total length of the canonical ranges on each level is $n$, the $O(\log n)$ levels account for the $O(n\log n)$ space usage.
In this section, we parameterize the number of levels by a trade-off parameter~$\tau$ to lower the number of levels while increasing the degree of the underlying interval hierarchy.
This gives a continuous trade-off between Kadane's linear-space scanning solution and the polylogarithmic-query data structure of \cref{thm:range_omss}.

Let $\tau$ be an integer with $1 \le \tau \le \ceil{\log_2 n}$, and define
\( \Delta := \ceil{n^{1/\tau}}. \)
We first introduce a rooted ordered interval tree $\mathcal T_\tau$.
The root corresponds to the range $[1..n]$.
For every node $v$ whose range $I_v=[g_v..h_v]$ has length $\ell_v:=h_v-g_v+1 \ge 2$, 
we partition $I_v$ into $\min(\Delta,\ell_v)$ consecutive ranges whose lengths differ by at most one, 
and make these ranges the children of $v$.
We stop after depth $\tau$ or when a range has length one.

\begin{lemma}\label{lem:tradeoff_hierarchy}
The tree $\mathcal T_\tau$ has depth at most $\tau$.
At every depth, its node ranges are pairwise disjoint and their total length is at most $n$.
\end{lemma}
\begin{proof}
The children of every node form a partition of the range of their parent.
Consequently, the node ranges at every depth are pairwise disjoint and have total length at most $n$.

After $q$ subdivision steps, the length of every non-singleton node range is at most
$\ceil{n/\Delta^q}$.
Since $\Delta^\tau \ge n$, every range has length at most one after $\tau$ subdivision steps.
Hence, the depth of $\mathcal T_\tau$ is at most $\tau$.
\end{proof}

\subsection{Node Summaries}\label{sec:tradeoff_node_summaries}

For every node $v$ of $\mathcal T_\tau$, we store a data structure that returns the MSS summary
$\mathcal S_\alpha(I_v)$ defined in \cref{sec:general}.
For $I_v=[g_v..h_v]$, let $\ell_v:=h_v-g_v+1$.
We store:
\begin{enumerate}[(1)]
    \item the length $\ell_v$ and the sum $S_0(g_v,h_v)$,
    \item a data structure for prefix OMSS queries on $X_0[g_v..h_v]$, \label{it:tradeoff:prefix}
    \item a data structure for suffix OMSS queries on $X_0[g_v..h_v]$, and \label{it:tradeoff:suffix}
    \item the OMSS data structure of \cref{lemma:sakaiDS} on $X_0[g_v..h_v]$.
\end{enumerate}

We describe (\ref{it:tradeoff:prefix}); the construction of (\ref{it:tradeoff:suffix}) is symmetric.
For every prefix length $t\in[1..\ell_v]$, let
\(
    p_v(t):=S_0(g_v,g_v+t-1).
\)
After subtracting the offset $\alpha$, the value of this prefix is $p_v(t)-\alpha t$.
Therefore, answering a prefix OMSS query is equivalent to finding a point maximizing $(x,y) \mapsto y-\alpha x$ among
\(
    P_v:=\{(t,p_v(t))\mid t\in[1..\ell_v]\},
\)
i.e. $\max_{t\in[1..\ell_v]} \{p_v(t)-\alpha t\}$.
As in the construction in \cref{sec:suffix}, we retain the vertices of the upper convex hull of $P_v$.
Since the points are sorted by their first coordinate, the hull can be constructed in $O(\ell_v)$ time and occupies $O(\ell_v)$ space.
The breakpoints between consecutive hull vertices are sorted, and hence a binary search returns the optimal prefix length in $O(\log\ell_v)$ time.
We obtain the suffix structure analogously from the points
\( \{(t,S_0(h_v-t+1,h_v))\mid t\in[1..\ell_v]\}. \)

\begin{lemma}\label{lem:tradeoff_node_summary}
For every node $v$ of $\mathcal T_\tau$, we can construct an $O(\ell_v)$-space data structure in
$O(\ell_v\log^2\ell_v)$ time such that we can obtain $\mathcal S_\alpha(I_v)$ in $O(\log\ell_v)$ time for every non-negative offset $\alpha \ge 0$.
\end{lemma}
\begin{proof}
We can compute the total score
\(
    S_\alpha(g_v,h_v)
    =
    S_0(g_v,h_v)-\alpha\ell_v
\)
in constant time.
The prefix and suffix structures described above each use $O(\ell_v)$ space, can be constructed in $O(\ell_v)$ time, and can be queried in $O(\log\ell_v)$ time.
Finally, \cref{lemma:sakaiDS} gives an $O(\ell_v)$-space OMSS data structure that can be constructed in
$O(\ell_v\log^2\ell_v)$ time and queried in $O(\log\ell_v)$ time.
Combining the four fields gives the stated bounds.
\end{proof}

\subsection{Range Decomposition and Queries}\label{sec:tradeoff_query}

Given a query range $[a..b]$, we compute a canonical cover with respect to $\mathcal T_\tau$ as follows.
Starting at the root, we discard a node whose range is disjoint from $[a..b]$, report a node whose range is contained in $[a..b]$, and recurse into all children of every remaining node.
The reported ranges are pairwise disjoint, occur in left-to-right order, and form a partition of $[a..b]$.

\begin{lemma}\label{lem:tradeoff_cover}
We can partition every range $[a..b]\subseteq[1..n]$ into $O(\Delta \tau)=O(\tau n^{1/\tau})$ node ranges of $\mathcal T_\tau$.
We can compute the partition in $O(\Delta \tau)$ time.
\end{lemma}
\begin{proof}
At every depth, at most two node ranges intersect $[a..b]$ without being contained in $[a..b]$: the node containing $a$ and the node containing $b$.
Every such node has at most $\Delta$ children.
Apart from at most two children that remain only partially covered, every intersecting child can be reported immediately.
Thus, we report at most $2 \Delta$ node on each level.
Since $\mathcal T_\tau$ has depth at most $\tau$ by \cref{lem:tradeoff_hierarchy}, the canonical cover contains $O(\Delta \tau)$ ranges.
The same traversal computes the cover in $O(\Delta \tau)$ time.
Since $\ceil{x}\le 2x$ for every $x\ge 1$, we have $\Delta=O(n^{1/\tau})$.
\end{proof}

Let $I_1,\dots,I_c$ be the canonical cover of $[a..b]$, ordered from left to right.
For every $k\in[1..c]$, we query the data structure of \cref{lem:tradeoff_node_summary} at the node representing $I_k$ and obtain $\mathcal S_\alpha(I_k)$.
We then merge these summaries from left to right using \cref{lem:mss_summary_merge}.
After processing $I_1,\dots,I_k$, the maintained summary is exactly
$\mathcal S_\alpha(I_1\cup\cdots\cup I_k)$.
Consequently, the $\mathsf{best}$ component after processing all $c$ ranges is an answer to the R-OMSS query on $[a..b]$.

\begin{theorem}\label{thm:range_omss_tradeoff}
Let $\tau$ be an integer with $1\le \tau\le\ceil{\log_2 n}$.
Given an array $X_0$ of size $n$, 
there exists an $O(\tau n)$-space data structure that can answer an R-OMSS query on $X_0$ in \( O\bigl(\tau n^{1/\tau}\log n\bigr) \)
time for every non-negative offset $\alpha \ge 0$.
We can construct this data structure in $O(\tau n\log^2 n)$ time.
\end{theorem}
\begin{proof}
By \cref{lem:tradeoff_hierarchy}, the total length of the node ranges at every depth is at most $n$.
By \cref{lem:tradeoff_node_summary}, the data structures at one depth therefore use $O(n)$ space and can be constructed in $O(n\log^2 n)$ time.
There are at most $\tau+1=O(\tau)$ depths, yielding $O(\tau n)$ space and $O(\tau n\log^2 n)$ construction time.

By \cref{lem:tradeoff_cover}, the query range has a canonical cover of size
$c=O(\tau n^{1/\tau})$.
Obtaining one node summary takes $O(\log n)$ time by \cref{lem:tradeoff_node_summary}, and merging two summaries takes constant time by \cref{lem:mss_summary_merge}.
The total query time is therefore
$O(c\log n)=O(\tau n^{1/\tau}\log n)$.
\end{proof}

For constant $\tau$, \cref{thm:range_omss_tradeoff} uses linear space and gives a sublinear query time.
In particular, we obtain the following consequence.

\begin{corollary}\label{cor:range_omss_linear_space}
For every fixed constant $\varepsilon>0$, there exists an $O(n)$-space data structure that can answer an R-OMSS query on $X_0$ in
\(
    O(n^\varepsilon\log n)
\)
time.
The data structure can be constructed in $O(n\log^2 n)$ time, where the constants hidden in the asymptotic notation depend on $\varepsilon$.
\end{corollary}
\begin{proof}
Set $\tau=\ceil{1/\varepsilon}$ in \cref{thm:range_omss_tradeoff}.
Then $\tau$ is a constant and $n^{1/\tau}\le n^\varepsilon$.
\end{proof}

The parameter $\tau$ continuously trades space for query time.
For example, setting
$\tau=\left\lceil\frac{\log n}{\log\log n}\right\rceil$
gives
\( O\!\left(\frac{n\log n}{\log\log n}\right) \)
space and 
    \( O\!\left(\frac{\log^3 n}{\log\log n}\right) \)
    query time.
At the other endpoint, $\tau=\ceil{\log_2 n}$ yields a binary hierarchy and recovers the
$O(n\log n)$-space and $O(\log^2 n)$-query bounds of \cref{thm:range_omss}.

\section{Special Types of Inputs}\label{sec:special}
Based on the results of the previous section, 
we analyze and devise specializations of our data structures for R-OMSS queries on $X_0$ under special input types. 
Specifically, we consider the following three scenarios: 
(1) the offset $\alpha$ is restricted to positive integers, 
(2) $X_0$ is an array over an alphabet of size $\sigma$, and 
(3) $X_0$ is given in run-length encoded (RLE) form.

\subsection{Integer Offsets}
When the offset $\alpha$ is restricted to positive integers, the data structure of \cref{thm:range_omss} can be improved by the following modifications. First, we introduce a lemma that states the data structure of \cref{lemma:sakaiDS} can be improved under this restriction.

\begin{lemma}\label{lem:sakai_integer}
When the offset $\alpha$ is restricted to positive integers, there exists an $s(n, U)$-space data structure that answers an OMSS query on $X_0$ in $t(n, U)$ time, where $s(n, U)$ and $t(n, U)$ denote the space and query-time complexities, respectively, of a predecessor data structure storing $n$ integers from a universe of size $U$.
\end{lemma}
\begin{proof}
    We modify the data structure of \cref{lemma:sakaiDS} as follows.
    For each (offset, interval) pair $(\alpha, I)$ in $L$ (the list defined in the statement of \cref{lemma:sakaiDS}), we replace it with the integer pair $(\ceil{\alpha}, I)$.
    If there exist two distinct pairs $(\alpha_1, I_1)$ and $(\alpha_2, I_2)$ in $L$ such that $\alpha_1 < \alpha_2$ and $\ceil{\alpha_1} = \ceil{\alpha_2}$, 
    we discard the pair $(\alpha_1, I_1)$ and maintain only $(\ceil{\alpha_2}, I_2)$ in $L$. 
    We also discard all pairs $(\alpha, I)$ with $\ceil{\alpha} > U$.
    We then answer the query by performing the binary search over the offsets by maintaining a predecessor data structure on $L$.
\end{proof}

For example, if we use a y-fast trie~\cite{willard1983log} as the predecessor data structure, for which $s(n, U) = O(n)$ and $t(n, U) = O(\log \log U)$, 
then an OMSS query on $X_0$ can be answered in $O(\log \log U)$ time using $O(n)$ space; this is $O(\log\log n)$ when $U=n^{O(1)}$.

Next, for each canonical range $[g..h]$ of $X_0$, let $T_{[g..h]}$ be the rooted subtree of the SBST $T$ on $X_0$ (defined in \cref{sec:suffix}) whose root corresponds to $[g..h]$.
Let $I_0[g..h]$ be the shortest suffix R-OMSS answer on $[g..h]$ at offset $0$.
We define a catalog $L_{[g..h]}$ containing the sentinel pair
\(
    (0,I_0[g..h])
\)
and, for every distinct positive threshold offset $\beta$ stored at a node of $T_{[g..h]}$, the pair
$(\beta,I_\beta[g..h])$, where $I_\beta[g..h]$ is the shortest suffix R-OMSS answer on the root range $[g..h]$ at offset $\beta$.
If the same positive threshold occurs at several nodes, we store only one pair for that key.
In particular, for $g=h$, the catalog consists only of $(0,[g..g])$.

The catalog has size $O(h-g+1)$.
To prove its correctness, write its distinct keys as
\(
    0=\beta_0<\beta_1<\cdots<\beta_k
\)
and set $\beta_{k+1}:=+\infty$.
For every $j\in[0..k]$ and every $\alpha\in[\beta_j,\beta_{j+1})$, every comparison with a node threshold made during the SBST descent has the same outcome as at offset $\beta_j$:
there is no stored node threshold strictly between $\beta_j$ and $\alpha$, and equality is routed to the right child.
Consequently, the descent reaches the same leaf, so the answer is $I_{\beta_j}[g..h]$.
This also proves correctness when $\alpha=\beta_j$, since routing equality to the right implements the right-active convention.
A descendant threshold that does not change the answer on the root range merely creates a redundant entry carrying the same root answer as an adjacent entry; keeping such thresholds therefore does not affect correctness.
The sentinel at offset $0$ ensures that every non-negative query offset has a predecessor.

Therefore, $L_{[g..h]}$ plays the same role as the data structure of \cref{lemma:sakaiDS} for suffix R-OMSS queries restricted to $[g..h]$.
By maintaining these catalogs for all canonical ranges of $X_0$ and applying the rounding and duplicate-removal modification from the proof of \cref{lem:sakai_integer}, we obtain the following lemma, since the total length of all canonical ranges of $X_0$ is $O(n \log n)$.

\begin{lemma}\label{lem:suffix_integer}
   When the offset $\alpha$ is restricted to positive integers and the query range is restricted to a canonical range of $X_0$, 
   there exists a data structure that can answer a suffix R-OMSS query on $X_0$ in $t(n \log n, U)$ time using $\sum_I s(|L_I|,U)$ space, where the sum is taken over all canonical ranges $I$ of $X_0$.
   Here, $s(n, U)$ and $t(n, U)$ denote the space and query-time complexities, respectively, of a predecessor data structure storing $n$ integers from a universe of size $U$.
\end{lemma}

Finally, by combining the data structures of \cref{lem:sakai_integer} and \cref{lem:suffix_integer}, together with $C_0$, we obtain the following corollary.

\begin{corollary}
Assume a word-RAM model with word size $w=\Omega(\log U)$. 
When the offset $\alpha$ is restricted to positive integers from the universe $[1..U]$, 
there exists an $O(n\log n)$-space data structure that answers an R-OMSS query on
$X_0$ in $O(\log n\log\log U)$ time. 
If $U=n^{O(1)}$, the query time is $O(\log n\log\log n)$.
\end{corollary}
\begin{proof}
We use y-fast tries as the predecessor structures of \cref{lem:sakai_integer,lem:suffix_integer} (and the symmetric prefix counterpart of \cref{lem:suffix_integer}).
Each y-fast trie is static, uses linear space in the size of its catalog, and answers a predecessor query in $O(\log\log U)$ worst-case time.
The total number of entries in all catalogs is $O(n\log n)$ 
because a catalog associated with a canonical range $I$ has $O(|I|)$ entries,
and the total length of all canonical ranges is $O(n\log n)$. 
Hence, the total space is $O(n\log n)$.
An R-OMSS query uses a canonical cover of size $O(\log n)$ and performs a constant number of predecessor queries for each range in the cover. 
The resulting summaries are merged in constant time per range. 
Thus, the total query time is $O(\log n\log\log U)$.
\end{proof}

\subsection{Bounded Alphabet Size}
While the number of distinct answers represented by the breakpoint list of \cref{lemma:sakaiDS} can be $\Theta(n)$,
this number can be less if $X_0$ is an array over an alphabet $\Sigma\subset\mathbb R$ of bounded size $\sigma$.
In particular, if $\sigma = 2$, we show that the size of this list is $\Theta(n^{2/3})$ in the worst case.

In this subsection, an \emph{answer range} means any range attaining the maximum value, before selecting a shortest answer.
We call two ranges $I$ and $J$ \emph{compatible} if
\(
    (|I|, S_0(I))=(|J|, S_0(J)).
\)
Equivalently, $S_\alpha(I)=S_\alpha(J)$ for every offset $\alpha$.
For a binary array, this is equivalent to saying that the two corresponding subarrays are permutations of each other.
Two offsets $\alpha$ and $\beta$ are called \emph{non-compatible} if no answer range of the OMSS query with offset $\alpha$ is compatible with an answer range of the OMSS query with offset $\beta$.

\subparagraph{Roadmap}
The basic idea is study inputs with a high number of non-compatible offsets.
Geometrically speaking, such non-compatible offsets emerge when the edges of the upper convex hull of the points have pairwise distinct slopes.
Such slope changes naturally occur in a strictly monotonically decreasing order when drawing a quarter of a circle from its leftmost to topmost point.
The main challenge is thus to find the correct discretization of the circle quarter to be expressible with binary alphabets.
In the setting of binary alphabets, Christoffel words are used in literature to express horizontal and vertical moves by 0 and 1, respectively.
In what follows, we combine some geometric discretization ideas with the strictly monotonically decreasing slopes expressible by a family of Christoffel words.
For that, we start with some definitions.

For every length $\ell\in[1..n]$, define
\[
    f_{X_0}(\ell)
    :=
    \max_{1\le i\le n-\ell+1}
    \sum_{j=i}^{i+\ell-1}X_0[j].
\]
Thus, $f_{X_0}(\ell)$ is the maximum number of ones in a length-$\ell$ subarray of $X_0$.
The value of an OMSS query with offset $\alpha$ is
\begin{equation}\label{eq:binary_omss_envelope}
    \max_{\ell\in[1..n]}
    \bigl(f_{X_0}(\ell)-\alpha\ell\bigr).
\end{equation}
Consequently, the compatible answer types that can occur for some offset are represented by exposed points of the upper convex hull of
\(
    P_{X_0}
    :=
    \{(\ell,f_{X_0}(\ell))\mid \ell\in[1..n]\}.
\)

We first bound the number of vertices of this hull.
A \emph{convex chain} is a list $p_0,\ldots,p_h$ of integer points in $[0..n]^2$ with increasing first coordinates and non-decreasing second coordinates
with the property that for the edges $e_i:=p_i-p_{i-1} = (a_i, b_i)$, $i\in[1..h]$, the slopes $\frac{b_i}{a_i}$ of the edges are strictly decreasing.
The \emph{primitive direction vector} of an edge $e_i$ is defined as
\(
\left(
\frac{a}{\gcd(|a|,|b|)},
\frac{b}{\gcd(|a|,|b|)}
\right).
\)
Edges have the same primitive direction vector if and only if they have the same slope.
On a strictly convex chain, two different edges cannot have the same slope.
This allows us to count edges by counting possible primitive directions:

\begin{lemma}\label{lem:binary_lattice_chain}
Let $p_0,\ldots,p_h$ be a convex chain. 
Then $h=O(n^{2/3})$.
\end{lemma}
\begin{proof}
For every $i\in[1..h]$, write
\(
    e_i := p_i-p_{i-1} := (a_i,b_i).
\)
By the definition of a convex chain, we have $a_i\ge1$ and $b_i\ge0$ for every $i\in[1..h]$, and thus
\( \sum_{i=1}^{h}a_i\le n \)
and
\( \sum_{i=1}^{h}b_i\le n. \)
Since consecutive edges of the chain have distinct slopes, their primitive direction vectors are pairwise distinct.
Fix an integer $R\ge1$.
There are only $O(R^2)$ primitive integer vectors $(a,b)$ with $a\ge1$, $b\ge0$, and $a+b\le R$.
Hence, the chain has at most $O(R^2)$ edges with $a_i+b_i\le R$.
On the other hand, the number of edges with $a_i+b_i>R$ is at most
\(
    \frac{1}{R}\sum_{i=1}^{h}(a_i+b_i)
    \le
    \frac{2n}{R}
\)
since each of them contributes more than $R$ to the sum $\sum_{i=1}^{h}(a_i+b_i)$.
Therefore,
\( h=O\left(R^2+\frac{n}{R}\right) \), for every $R\ge1$.
Choosing $R=\ceil{n^{1/3}}$ proves the claim.
\end{proof}

Alternatively, we can show \cref{lem:binary_lattice_chain} by Jarník's theorem (\cite[Satz~1]{jarnik26uber} or \cite{barany12jarnik}) that states that the perimeter of a convex lattice polygon with $h$ vertices is $\Omega(h^{3/2})$ under every fixed norm
such as the Manhattan distance we used in the proof.
Since the perimeter in $[1..n] \times [1..n]$ is at most $4n$, $c h^{3/2} \le 4n$ for some constant $c$, and thus $h = O(n^{2/3})$.

\begin{corollary}\label{cor:bounded_alphabet}
The maximum number of pairwise non-compatible offsets of a binary array of size $n$ is $O(n^{2/3})$.
\end{corollary}
\begin{proof}
    The function $f_{X_0}$ is non-decreasing: a length-$\ell$ subarray attaining $f_{X_0}(\ell)$ can be extended by one position to a length-$(\ell+1)$ subarray without decreasing its number of ones.
    Hence, the vertices of the upper convex hull of $P_{X_0}$ form a convex lattice chain whose two coordinates are non-decreasing.
    By \cref{lem:binary_lattice_chain}, this chain has $O(n^{2/3})$ vertices.
    
    Let $A$ be a set of pairwise non-compatible offsets.
    For every $\alpha\in A$, the linear function $(x,y)\mapsto y-\alpha x$ exposes a face of the upper hull of $P_{X_0}$: 
    this face is a single vertex except at a breakpoint, where it is the edge joining two consecutive hull vertices.

    Choose one vertex $v_\alpha$ of this face.
    The point $v_\alpha$ represents an answer type of the OMSS query with offset $\alpha$ by \cref{eq:binary_omss_envelope}.
    If $\alpha,\beta\in A$ are distinct, then $v_\alpha\ne v_\beta$, since otherwise the two offsets would have compatible answer ranges.
    Thus, $\alpha\mapsto v_\alpha$ is injective, and $|A|=O(n^{2/3})$.
\end{proof}

For the lower bound, we use the following definitions.
A lattice point $(q,p)\in\mathbb Z^2$ is called \emph{primitive} if $\gcd(|q|,|p|)=1$.
Let
\(
    \mathbb P
    :=
    \left\{
        (q,p)\in\mathbb Z^2
        \ \middle|\
        \gcd(p,q)=1
    \right\}
\)
denote the set of primitive lattice points.
A polygon is called \emph{rational} if all its vertices have rational coordinates.  
Further, a polygon~$\mathcal A$ is called \emph{star-shaped} (with respect to the origin) if,
for every $(x,y)\in\mathcal A$ and every $\lambda\in[0,1]$, we have
$\lambda(x,y) = (\lambda x,\lambda y)\in\mathcal A$.
We use the following lemma.

\begin{lemma}[{\cite[Eq.(1.2)]{barany20primitive}}]\label{lem:primitive_lattice_estimate}
For a bounded rational polygon $\mathcal A$ that is star-shaped with
respect to the origin, 
\[
    \left|
        t\mathcal A \cap \mathbb P 
    \right|
    =
    \frac{6}{\pi^2}
    \operatorname{area}(\mathcal A)t^2
    +o(t^2)
    \text{~as~}
t\to\infty,
\]
where $\pi$ is the ratio of the circumference of a circle to its diameter.
\end{lemma}

\begin{figure}[t]
    \centering
    \includegraphics[width=0.8\textwidth]{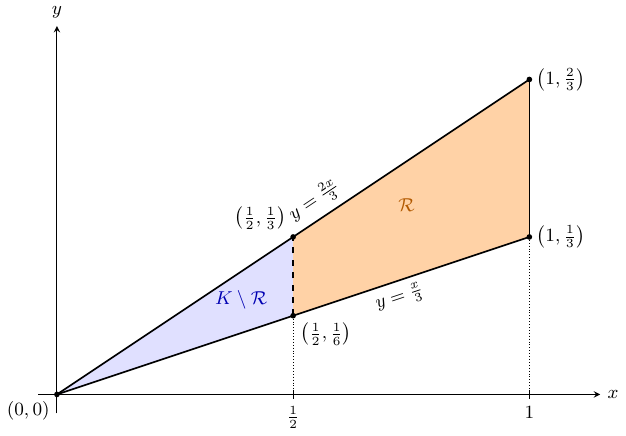}
    \caption{The triangular region $\mathcal K$ and the trapezoidal region $\mathcal R$
    used in the proof of \cref{thm:bounded_alphabet}.
    $\mathcal K$ covers the region $\mathcal K \setminus \mathcal R$ and $\mathcal R$.
}
    \label{fig:bounded_alphabet}
\end{figure}

The main idea of proving the following theorem is to build a binary string by concatenating binary strings of decreasing densities of ones and appending a sufficiently long run of ones.
These binary strings follow the spirit of lower Christoffel words for approximating line segments~\cite{luca25digital}.
We choose the densities to be rational numbers and pairwise distinct,
which we obtain from the primitive lattice points in a trapezoidal region, cf.~\cref{fig:bounded_alphabet}.

\begin{theorem}\label{thm:bounded_alphabet}
The maximum number of pairwise non-compatible offsets of a binary array of size $n$ is $\Theta(n^{2/3})$.
\end{theorem}
\begin{proof}
We give a matching lower bound.

Fix an integer $m\ge1$, and let
\[
    \mathcal R
    :=
    \left\{
        (x,y)\in\mathbb R^2
        \ \middle|\
	x \in \left[\frac12, 1\right] 
	\wedge
	y \in \left[\frac{x}{3},\frac{2x}{3}\right]
    \right\}.
\]

Define
\[
    \mathcal V_m
    :=
    m\mathcal R\cap\mathbb P
    =
    \left\{
        (x,y)\in\mathbb N^2
        \ \middle|\
	x \in \left[\ceil{\frac{m}{2}}, m\right]
	\wedge
	y \in \left[\ceil{x/3},\floor{2x/3}\right]
	\wedge
        \gcd(x,y)=1
    \right\}.
\]

Let
\(
    \mathcal K
    :=
    \left\{
        (x,y)\in\mathbb R^2
        \ \middle|\
	x \in [0,1]
	\wedge
	y \in [\frac{x}{3}, \frac{2x}{3}]
    \right\}.
\)
be the triangle with vertices
$ (0,0),$
$\left(1,\frac13\right),$ and
 $   \left(1,\frac23\right)$.
In particular $\mathcal K$ is a rational star-shaped polygon.
See \cref{fig:bounded_alphabet} for an illustration of $\mathcal K$ and $\mathcal R$.

Let $P(t)$ denote the number of primitive lattice points in $t\mathcal K$, i.e., $P = |\mathbb P \cap t\mathcal K|$.
Since
\[
    \operatorname{area}(\mathcal K)
    =
    \int_0^1
        \left(\frac{2x}{3}-\frac{x}{3}\right)
    \,dx
    =
    \frac16,
\]
\cref{lem:primitive_lattice_estimate} gives
\begin{equation}\label{eq:primitive_points_in_wedge}
    P(t)
    =
    \frac{1}{\pi^2}t^2+o(t^2).
\end{equation}

The set $\mathcal V_m$ consists of the primitive points in
$m\mathcal K$ whose first coordinate is at least $m/2$.
Consequently,
\(
    |\mathcal V_m|
    =
    P(m)-P(m/2)+O(m).
\)
The additive $O(m)$ term accounts for lattice points on the vertical boundary
line with x-coordinate $m/2$, and for the rounding in the inequalities defining $\mathcal V_m$.
Using \cref{eq:primitive_points_in_wedge}, we obtain
\begin{align}
    |\mathcal V_m|
    =
    \frac{1}{\pi^2}m^2
    -
    \frac{1}{\pi^2}\left(\frac m2\right)^2
    +o(m^2)
    +O(m)
    =
    \frac{3}{4\pi^2}m^2+o(m^2)
    =
    \Theta(m^2).
    \label{eq:primitive_pairs}
\end{align}

Since every $(q,p)\in\mathcal V_m$ satisfies
\(
    \frac{m}{2}\le q\le m,
\)
we obtain
\begin{equation}\label{eq:primitive_pair_length}
    N
    :=
    \sum_{(q,p)\in\mathcal V_m}q
    =
    \Theta(m^3),
\end{equation}
since
\(
    \frac{m}{2}|\mathcal V_m|
    \le
    N
    \le
    m|\mathcal V_m|
\)
and due to \cref{eq:primitive_pairs}.

For every $(q,p)\in\mathcal V_m$, define the binary string $W_{q,p}[1..q]$ by
\begin{equation}\label{eq:balanced_binary_block}
    W_{q,p}[s]
    :=
    \floor{\frac{sp}{q}}
    -
    \floor{\frac{(s-1)p}{q}}
    \text{~for each~}
    s\in[1..q].
\end{equation}
In the following, we call $W_{q,p}$ a \emph{block}.
The block $W_{q,p}$ contains exactly $p$ ones.
Also, every prefix of length $s$ contains
\(
    \sum_{j=1}^{s}W_{q,p}[j]
    =
    \floor{\frac{sp}{q}}
    \le
    \frac{sp}{q}
\)
ones.

Let us write
\(
    \mathcal V_m
    =
    \{(q_1,p_1),\ldots,(q_M,p_M)\}
\)
such that
\( \delta_1>\delta_2>\cdots>\delta_M, \)
with densities
\( \delta_i:=\frac{p_i}{q_i}\) for some $M \in \Theta(m^2)$.
The densities $\delta_i$ are pairwise distinct because all fractions $p_i/q_i$ are
reduced.
Let
\(
    W
    :=
    W_{q_1,p_1}
    W_{q_2,p_2}
    \cdots
    W_{q_M,p_M}
\)
and define our input array
\(
    X^{(m)}
    :=
    1^{2N}W
\)
for the OMSS problem.
By \cref{eq:primitive_pair_length}, the length of $X^{(m)}$ is
\(
    |X^{(m)}| = 3N=\Theta(m^3).
\)

For every $t\in[1..M-1]$, choose an offset
\(
    \alpha_t\in(\delta_{t+1},\delta_t).
\)
We claim that the unique OMSS answer on $X^{(m)}$ for $\alpha_t$ is the prefix
\begin{equation}\label{eq:binary_lower_answer}
    I_t
    :=
    \left[
        1..\,
        2N+\sum_{i=1}^{t}q_i
    \right].
\end{equation}

Since $\delta \in [\frac13,\frac23]$,
we have
$\alpha_t \in (0,\frac23)$.
Every nonempty prefix of $X^{(m)}_{\alpha_t}$ has positive sum.
Because $\alpha_t < 1$,
this is immediate for prefixes contained in the leading run of ones.
A prefix ending after $s\le N$ positions of $W$ has sum at least
\(
    2N(1-\alpha_t)-\alpha_t s
    \ge
    N(2-3\alpha_t)
    >
    0.
\)
Therefore, every maximum-sum segment starts at the first position:
we can extend a non-prefix substring to the left by the preceding nonempty prefix to strictly increase its sum.

Consider a prefix of length $s$ of the block $W_{q_i,p_i}$.
Its score is
\(
W_{q_i,p_i}.S_\alpha(1,s) :=
    \floor{\frac{sp_i}{q_i}}-\alpha_t s
    \le
    (\delta_i-\alpha_t)s.
\)
We consider two cases.

\subparagraph{Case 1: $i\le t$}
Then $\delta_i>\alpha_t$.
For every proper prefix $s<q_i$,
\[
W_{q_i,p_i}.S_\alpha(1,s) =
    \floor{\frac{sp_i}{q_i}}-\alpha_t s
    \le
    (\delta_i-\alpha_t)s
    <
    (\delta_i-\alpha_t)q_i
    =
    p_i-\alpha_t q_i 
    = 
    W_{q_i,p_i}.S_\alpha(1,q_i)
\]
Hence, within each of the first $t$ blocks, the prefix sum is uniquely
maximized at the end of each block.

\subparagraph{Case 2: $i>t$}
Then $\delta_i<\alpha_t$, and every nonempty prefix of $W_{q_i,p_i}$
has score
\[
W_{q_i,p_i}.S_\alpha(1,s) =
    \floor{\frac{sp_i}{q_i}}-\alpha_t s
    \le
    (\delta_i-\alpha_t)s
    <
    0.
\]
Thus, the prefix sums of $X^{(m)}_{\alpha_t}$ attain their unique maximum
at the boundary following the first $t$ blocks.
This proves \cref{eq:binary_lower_answer}.

The ranges
\(
    I_1,I_2,\ldots,I_{M-1}
\)
have pairwise distinct lengths, and each is the unique answer for its
respective offset.
Therefore, the offsets
\(
    \alpha_1,\alpha_2,\ldots,\alpha_{M-1}
\)
are pairwise non-compatible.
By \cref{eq:primitive_pairs,eq:primitive_pair_length},
\(
    M-1
    =
    \Theta(m^2)
    =
    \Theta\bigl(|X^{(m)}|^{2/3}\bigr).
\)

For an arbitrary target length $n$, choose
\(
    m
    :=
    \floor{(n/3)^{1/3}}.
\)
Since $N\le m^3$, the constructed array has length $3N\le n$.
Append zeros until its length is exactly $n$.
Every selected offset is positive, so the appended zeros have negative
adjusted value and do not change any of the unique answer ranges above.
Hence, there exists a binary array of length $n$ with
$\Omega(n^{2/3})$ pairwise non-compatible offsets.

\end{proof}

\subsection{Run-Length Encoded (RLE) Input}\label{sec:rle}
Let
\(
    X_0 := x_1^{\ell_1}x_2^{\ell_2}\cdots x_z^{\ell_z}
\)
be the run-length encoding of $X_0$, where $x_i\neq x_{i+1}$ for every $i \in [1..z-1]$ and $\ell_i\ge 1$ for every $i \in [1..z]$.
We define
\( L[i]:=\sum_{t=1}^{i}\ell_t \)
and
\( V[i]:=\sum_{t=1}^{i}\ell_t x_t \)
for $i\in[1..z]$, and stipulate that $L[0]=V[0]=0$.
Thus, the $i$-th run occupies the range $[L[i-1]+1..L[i]]$.
The \emph{run-value array} $[x_1,\ldots,x_z]$, the \emph{run-exponent array} $[\ell_1,\ldots,\ell_z]$, $L$, and $V$ use $O(z)$ space.
Given a position $p\in[1..n]$, its run can be found by a predecessor search in $L$ in $O(\log z)$ time, 
while a run boundary can be translated back to a position of $X_0$ in constant time.

\begin{example}\label{ex:rle}
    For $X_0 = [2^1, (-1)^{3}, 2^{2}]$ and $\alpha=0$, 
the MSS of the run-value array $[2,-1,2]$ is the entire array, 
whereas the MSS of $X_0$ is its last run.
\end{example}
In what follows, we want to obtain the space and time complexities in terms of $z$ for answering R-OMSS queries on the $O(z)$-space representation of $X_0$.
We therefore introduce a weighted version of the OMSS problem.

We call 
\(
    \widehat X  := ((x_1,\ell_1),\ldots,(x_z,\ell_z)),
\)
the \emph{weighted array} of $X_0$ and define
\(
    \widehat S_\alpha(i,j)
    :=
    \sum_{t=i}^{j} \ell_t(x_t-\alpha).
\)

A \emph{weighted positive OMSS} (WPOMSS) \emph{query} returns an interval maximizing $\widehat S_\alpha(i,j)$ if the maximum value is positive, and returns the empty interval otherwise.
Among all intervals attaining the maximum value, it returns one with the shortest expanded length $\lambda(i,j)=L[j]-L[i-1]$.
While implicit, the data structure of \cite{sakai2024data} allows us to compute WPOMSS queries as follows. 
We give proof details in \cref{app:thmWOPMSS}.

\begin{restatable}{theorem}{thmWOPMSS}
\label{thmWOPMSS}
Given a weighted array $\widehat X$ of size $z$, there exists an $O(z)$-space data structure that supports WPOMSS queries in $O(\log z)$ time.
We can construct the data structure in $O(z\log^2z)$ time.
\end{restatable}

Similarly, we can extend our suffix R-OMSS data structure of \cref{sec:suffix} to weighted arrays.
For that, we replace the identity
\(
    S_\alpha(i,j)
    =
    S_\beta(i,j)+(\beta-\alpha)(j-i+1)
\)
in \cref{prop:sbst}, 
by
\(
    \widehat S_\alpha(i,j)
    =
    \widehat S_\beta(i,j)
    +(\beta-\alpha)\bigl(L[j]-L[i-1]\bigr).
\)
Since $L[j]-L[i-1]>0$, the same monotonicity argument applies.
Moreover, the points used for constructing a threshold offset for suffixes ending at run $h$ become
\(
    \bigl(L[h]-L[h-t],\,V[h]-V[h-t]\bigr),
\)
whose first coordinates are strictly increasing in $t$.
Therefore, we can construct their upper hull in linear time, and obtain the following claim.

\begin{lemma}\label{lem:weighted_suffix}
Given a weighted array of size $m$, there exists an $O(m)$-space data structure supporting weighted prefix and suffix R-OMSS queries in $O(\log^2m)$ time.
We can construct the data structure in $O(m\log^2 m)$ time.
For a canonical query range, the query time is $O(\log m)$.
\end{lemma}
\begin{proof}
Apply the construction and proof of \cref{thm:range_suffix_omss} with the weighted sum $\widehat S_\alpha$.
Our aforementioned observations give the weighted counterparts of \cref{prop:sbst,prop:sbst2} and of the threshold-offset construction.
All remaining arguments use only additivity of interval sums and are unchanged.
\end{proof}

We next relate weighted run intervals to segments of the expanded array.
All answers obey the shortest-answer rule of \cref{eq:shortest_answer}.

\begin{lemma}\label{lem:rle_aligned_summary}
Let $p\le q$, and consider the run-aligned range
\(
    I=[L[p-1]+1..L[q]]
\)
of $X_0$.
Its MSS summary for an offset $\alpha$ can be obtained from the following information:
\begin{enumerate}[(1)]
    \item the total weighted sum
    \(
    \widehat S_\alpha(p,q) =
        \sum_{i=p}^{q}\ell_i(x_i-\alpha);
    \)
    \item a maximum nonempty weighted prefix~$\mathsf{pre}_\alpha(I)$ and suffix~$\mathsf{suf}_\alpha(I)$ of the run range $[p..q]$;
    \item a WPOMSS answer on $[p..q]$, if one exists; and
    \item a position $h\in[p..q]$ maximizing $x_h = \max_{r \in [p..q]} x_r$. \label{it:rle:max_run}
\end{enumerate}
More precisely, the best prefix is the better of the maximum weighted full-run prefix and the singleton $[L[p-1]+1..L[p-1]+1]$.
The best suffix is the better of the maximum weighted full-run suffix and the singleton $[L[q]..L[q]]$.
If a WPOMSS answer exists, it is an MSS of $X_\alpha[I]$; otherwise,
$[L[h-1]+1..L[h-1]+1]$ is an MSS of $X_\alpha[I]$.
\end{lemma}
\begin{proof}
Write $y_i:=x_i-\alpha$.
Consider first an answer range of positive value.
Assume that its left endpoint lies strictly inside the $i$-th run, and consider the following cases:
\begin{itemize}
    \item If $y_i>0$, extending the range to the left boundary of the run increases its value.
    \item If $y_i<0$, deleting the covered part of that run increases its value.
    \item If $y_i=0$, deleting that part preserves the value and gives a shorter range.
\end{itemize}
Thus, under the shortest-answer rule (\cref{eq:shortest_answer}), the left endpoint lies on a run boundary.
The argument for the right endpoint is symmetric.
Consequently, every positive MSS is a concatenation of complete runs and is returned by the WPOMSS data structure.

If no complete-run interval has positive value, 
then $y_i\le0$ for every $i\in[p..q]$, since otherwise the complete $i$-th run would have positive value.
In this case, 
every segment has non-positive value, 
and a shortest maximum segment consists of a single $X_0$-entry, 
namely one occurrence of a run value maximizing $y_i$, equivalently maximizing $x_i$.
This proves the statement for the unrestricted MSS.

For a prefix, suppose that its right endpoint lies strictly inside the $i$-th run.
If $i>p$, then extending to the right run boundary improves the value when $y_i>0$, while deleting the covered part of the run improves the value when $y_i<0$ and preserves the value with a shorter range when $y_i=0$.
Hence, an optimal prefix either ends on a run boundary or is contained in the first run.
The latter case is represented by the first singleton; if $y_p>0$, the complete first run is at least as good.
This proves the prefix statement, and the suffix statement is symmetric.
\end{proof}

For computing (\ref{it:rle:max_run}) in \cref{lem:rle_aligned_summary}, 
we maintain a standard $O(z)$-space range-maximum-query data structure \textsf{RMQ} on $x_1,\ldots,x_z$ that returns an arbitrary maximum in constant time~\cite{gabow84scaling}.
Together with \cref{thmWOPMSS,lem:weighted_suffix}, this gives an $O(m)$-space parametric MSS summary for every range of $m$ complete runs.
The summary consists of the same four fields as in \cref{sec:general}, and two adjacent summaries are combined by \cref{eq:summary_total,eq:summary_prefix,eq:summary_suffix,eq:summary_best}.

For a subrange of consecutive positions inside one run of value $x$, we can compute its summary in constant time since all entries are equal.
If $x-\alpha > 0$, its best prefix, suffix, and unrestricted segment are the entire range.
Otherwise, they are respectively its first position, its last position, and its first position.

\begin{lemma}\label{lem:sakai_rle}
There exists an $O(z)$-space data structure that supports OMSS queries on $X_0$ in $O(\log z)$ time.
It can be constructed in $O(z \log^2 z)$ time.
\end{lemma}
\begin{proof}
Apply \cref{thmWOPMSS} to the weighted run array
\(
    \widehat X = ((x_1,\ell_1),\ldots,(x_z,\ell_z)).
\)
If the query returns a nonempty run range $[p..q]$, then
\(
    [L[p-1]+1..L[q]]
\)
is an OMSS answer by \cref{lem:rle_aligned_summary}.
If the WPOMSS query returns the empty range, we return one single $X_0$-entry of a run whose value $x_i$ is maximum.
This run is obtained from \textsf{RMQ}.
\end{proof}

\begin{corollary}\label{cor:rangesuffix_rle}
There exists an $O(z)$-space data structure that supports suffix R-OMSS queries on $X_0$ in $O(\log^2 z)$ time.
It can be constructed in $O(z \log^2 z)$ time.
\end{corollary}
\begin{proof}
Use the weighted suffix R-OMSS data structure of \cref{lem:weighted_suffix} on the run array.
Given a query range $[a..b]$, let $p$ and $q$ be the runs containing $a$ and $b$, respectively.
If $p=q$, the answer is obtained directly from the constant run $X_0[a..b]$.
Otherwise, write the query range as the concatenation of
\( [a..L[p]],\)
 \(   [L[p]+1..L[q-1]],\) 
 and
 \(   [L[q-1]+1..b], \)
omitting the middle range when $q=p+1$.
The first and last ranges are subranges of single runs and have summaries computable in constant time.
For the middle range, its total is obtained from $V$ and $L$, and its suffix is obtained from the weighted suffix R-OMSS query on the run range $[p+1..q-1]$, together with the singleton correction of \cref{lem:rle_aligned_summary}.
Merging the at most three summaries from left to right yields the suffix of $X_\alpha[a..b]$.
Locating $p$ and $q$ takes $O(\log z)$ time, and the weighted suffix R-OMSS query takes $O(\log^2z)$ time.
\end{proof}

The time--space trade-off of \cref{sec:tradeoff} can be applied to the sequence of runs.
At a node containing $m$ runs, we store weighted prefix and suffix R-OMSS catalogs and the WPOMSS data structure of \cref{thmWOPMSS}.
By \cref{lem:rle_aligned_summary}, \textsf{RMQ} supplies the possible singleton answer, so the node stores an exact MSS summary of the corresponding expanded run-aligned range in $O(m)$ space.
The proof of \cref{thm:range_omss_tradeoff} then applies without further changes.
For a query in the original array~$X_0$, we additionally merge the two boundary subranges described in the proof of \cref{cor:rangesuffix_rle}.

\begin{corollary}\label{cor:range_rle_tradeoff}
Let $\tau$ be an integer with $1\le \tau\le\ceil{\log_2z}$.
There exists an $O(\tau z)$-space data structure that supports R-OMSS queries on $X_0$ in
\( O\bigl(\tau z^{1/\tau}\log z\bigr) \)
time.
It can be constructed in $O(\tau z\log^2z)$ time.
In particular, for every fixed $\varepsilon>0$, R-OMSS queries can be supported in
$O(z^\varepsilon\log z)$ time using $O(z)$ space.
\end{corollary}

\begin{remark}\label{remark:range_rle}
When the offset is restricted to an integer universe $[0..U-1]$, the breakpoints in the (weighted) catalogs can be rounded and stored in integer predecessor structures exactly as in \cref{lem:sakai_integer}.
The run containing a query position is found by a predecessor query in $L\subseteq[1..n]$.
Thus, the corresponding bounds are obtained by substituting the query time of the chosen predecessor structure for the binary searches on offset catalogs and adding the predecessor time for $L$.
For example, a y-fast trie gives $O(\log\log U)$ time for an offset lookup and $O(\log\log n)$ time for locating a run, while using linear space in the number of stored keys.
\end{remark}

\section{Lower Bounds}\label{sec:lb}
Finally, we present lower bound results for data structures that answer R-OMSS queries on $X_0$. 
We first show a reduction from the predecessor query to the OMSS query,
and subsequently
show that any encoding data structure that answers suffix R-OMSS queries requires at least $\Omega(n \log n)$ bits.

\subparagraph*{Time-Space Trade-Offs Lower Bound.}
The following theorem gives a reduction from predecessor queries to OMSS queries. 

\begin{theorem}\label{thm:omss_predtradeoff}
    For any set $S$ of $n$ values from a universe $U \subset \mathbb R$, 
    given a value $\alpha \in U$, 
    the predecessor of $\alpha$ in $S$ can be determined by an OMSS query and at most two accesses to an array of size $n$.
\end{theorem}
\begin{proof}
   Suppose we want to answer a predecessor query for a value $\alpha \in U$ in $S$.
   For that, 
   we have constructed, in a precomputation step, an OMSS data structure on $X_0$, 
   where $X_0$ stores the values of $S$ in increasing order,
   and compute the OMSS query with offset $\alpha$.
   If the answer range for $\alpha$ is $[n..n]$, 
   then the predecessor of $\alpha$ is either $X_0[n]$ (if $X_0[n] \le \alpha$) or $X_0[n-1]$ (if $X_0[n] > \alpha$).
   If the answer range is $[1..n]$, then there is no predecessor of $\alpha$ in $S$.
   Otherwise, the answer range is $[i..n]$ with $i \in [2..n-1]$.
   Then the predecessor of $\alpha$ is $X_0[i-1]$ 
   since 
   (a) $X_{\alpha}[k] < 0$ for all $k \le i-2$,
   (b) $X_{\alpha}[i-1] \le 0$,  and
   (c) $X_{\alpha}[k] > 0$ for all $k \ge i$.
   By \cref{eq:shortest_answer}, zero-valued entries adjacent to the positive suffix
   are omitted, so the returned range starts at the first position whose value is strictly
   larger than $\alpha$.
\end{proof}

The conclusion inherited from \cref{thm:omss_predtradeoff} depends on the predecessor model.
In the comparison model over arbitrary real keys, a predecessor query has $n+1$ possible answers.
Hence, its comparison decision tree has at least $n+1$ leaves and depth at least
$\ceil{\log_2(n+1)}$.
Since the reduction uses one OMSS query followed by only $O(1)$ additional comparisons and array accesses, every comparison-based OMSS index requires $\Omega(\log n)$ query time.
The same lower bound applies to suffix R-OMSS because, 
on the sorted instances used in \cref{thm:omss_predtradeoff},
every OMSS answer is a suffix of $[1..n]$. 
It applies also to R-OMSS because OMSS is a special case of R-OMSS.

For integer keys from $[0..2^\ell-1]$, let the cell size be $w \ge \ell$ bits.
If an OMSS index uses $S$ cells and answers a query in $t$ cell probes, then the reduction of \cref{thm:omss_predtradeoff}, together with the sorted input array, yields a predecessor structure using $S+O(n)$ cells and $t+O(1)$ probes.
Thus, writing $T_{\mathrm{pred}}(n,\ell,w,S)$ for the Pătraşcu and Thorup's predecessor lower bound~\cite{puatracscu2006time}, 
we obtain
\(
    t+O(1)
    \ge
    T_{\mathrm{pred}}(n,\ell,w,S+O(n)).
\)
This is a parameter-dependent trade-off and does not imply an unconditional $\Omega(\log n)$ lower bound in the word-RAM or cell-probe model.
For example, for polynomial universes, where $\ell=\Theta(\log n)$, and near-linear space, predecessor search can be supported in $\Theta(\log\ell)=\Theta(\log\log n)$ time.

\subparagraph*{Space Lower Bounds in the Encoding Model. }
One may also study the R-OMSS query in the variant where only the answer range in $X_0$ is returned.
In this setting, the problem is considered under the \emph{encoding model}: after preprocessing $X_0$, we maintain only the information necessary to answer queries, without storing $X_0$ explicitly.
For many range queries, encoding data structures use less space than explicitly storing $X_0$; see~\cite{raman2015encoding} for an overview.
For example, there exists an $O(n)$-bit encoding for range MSS queries under a word-RAM model with word size $\Theta(\log n)$ bits~\cite{gawrychowski2015encodings}, whereas explicitly storing an array over a universe of size $n$ requires $\Theta(n\log n)$ bits.

The following elementary observation shows that suffix R-OMSS queries can simulate access to the input array.

\begin{lemma}\label{lem:suffix_r_omss_access}
Let $i\in[1..n-1]$ and let $\alpha\ge0$.
For a suffix R-OMSS query on the range $[i..i+1]$, 
the shortest answer range is
\( [i..i+1] \)
if and only if
\( X_0[i]>\alpha. \)
Otherwise, the shortest answer range is $[i+1..i+1]$.
\end{lemma}
\begin{proof}
The only two suffixes of $[i..i+1]$ are $[i..i+1]$ and $[i+1..i+1]$.
Their scores are
\(S_\alpha([i..i+1])=X_0[i]+X_0[i+1]-2\alpha\)
and
\(S_\alpha([i+1..i+1])= X_0[i+1]-\alpha\),
respectively.
The first value exceeds the second one if and only if $X_0[i]>\alpha$.
If $X_0[i]=\alpha$, the two values are equal, and
\cref{eq:shortest_answer} selects the shorter range $[i+1..i+1]$.
\end{proof}

We now prove that an encoding for suffix R-OMSS queries contains enough information to reconstruct every permutation.

\begin{theorem}\label{thm:range_omss_lb}
Any permutation $X_0$ of $[1..n]$ can be reconstructed from an encoding that supports suffix R-OMSS queries.
Consequently, every such encoding requires at least
\( \lg(n!)=\Omega(n\log n) \) bits.
\end{theorem}
\begin{proof}
Fix a position $i\in[1..n-1]$.
By \cref{lem:suffix_r_omss_access}, a query on $[i..i+1]$ with offset
\( \alpha=k+\frac12 \)
determines whether $X_0[i]\ge k+1$.
A binary search over $k\in[0..n-1]$ therefore determines $X_0[i]$ using
$O(\log n)$ suffix R-OMSS queries.

Applying this procedure to every $i\in[1..n-1]$ reconstructs
$X_0[1..n-1]$.
Since $X_0$ is a permutation of $[1..n]$, its last entry is the unique value of $[1..n]$
that has not appeared among the first $n-1$ entries.
Hence, the complete permutation can be reconstructed from the encoding.

There are $n!$ permutations of $[1..n]$.
Two distinct permutations must therefore have distinct encodings, and any encoding requires
at least $\lg(n!)=\Omega(n\log n)$ bits in the worst case.
\end{proof}

Since suffix R-OMSS is a special case of R-OMSS, the same lower bound applies to encodings for general R-OMSS queries.
In particular, the $O(n)$-word space bound of \cref{thm:range_suffix_omss} is optimal up to a constant factor in the word-RAM model with word size $\Theta(\log n)$.

The same access argument gives the corresponding lower bound for bounded alphabets.

\begin{corollary}\label{cor:range_omss_boundedalphabet}
For every $\sigma\ge2$, any encoding data structure that supports suffix R-OMSS queries on arrays of length $n$ over an alphabet of size $\sigma$ requires
\( \Omega(n\log\sigma) \)
bits.
The lower bound also applies to R-OMSS queries.
\end{corollary}
\begin{proof}
    \Cref{lem:suffix_r_omss_access} allows us to reconstruct an input up to the last letter.
    For full reconstruction, it therefore suffices to fix the last letter.
    So given an alphabet $\Sigma := [1..\sigma]$,
let
\(
    \mathcal X
    :=
    [1..\sigma]^{n-1}\times\{1\}
\)
be the set of all strings of length $n$ build on $\Sigma$ and ending with the letter~$1$.
For every $X_0\in\mathcal X$ and every $i\in[1..n-1]$,
\cref{lem:suffix_r_omss_access} and a binary search with half-integer offsets reconstruct $X_0[i]$ using $O(\log\sigma)$ queries.
Therefore, we can reconstruct $X_0$ from its encoding.
Since
\(
    |\mathcal X|=\sigma^{n-1},
\)
the encoding must use at least
\(
    \log|\mathcal X|
    =
    (n-1)\log\sigma
    =
    \Omega(n\log\sigma)
\)
bits.
\end{proof}

\cref{cor:range_omss_boundedalphabet} shows that, also for bounded alphabets, suffix R-OMSS and R-OMSS encodings require asymptotically as much space as the input itself.

\section{Open Problems}

There are many gaps left open regarding time and space complexities.
It is puzzling to us whether we can construct the data structure of \cref{lemma:sakaiDS} in optimal $O(n)$ time.
This would improve the construction time of our R-OMSS data structures in \cref{thm:range_omss}.
Here, we additionally want to improve the space to $O(n)$, which would be optimal by \cref{thm:range_omss_lb}.
Improving the query time to $O(\log n)$ is also an interesting open problem.
Such a bound would match the lower bound transferred by \cref{thm:omss_predtradeoff} in the comparison model over arbitrary real offsets.
For bounded integer universes in the word-RAM or cell-probe model, the inherited target instead depends on the universe size, word size, and space usage.
In case that the number of distinct values in $X_0$ is bounded, 
\cref{thm:bounded_alphabet} settles the maximum number of pairwise non-compatible offsets.
Determining the corresponding tight bound for larger fixed alphabets over arbitrary real values remains open.

For future directions, we want to study the variations of the MSS problem described in the introduction as query problems with offset and/or query range parameters.

\bibliography{refs}

\clearpage
\appendix
\section{Proof of \cref{thmWOPMSS}}\label{app:thmWOPMSS}

We here prove \cref{thmWOPMSS} of \cref{sec:rle} by following Sakai's construction~\cite{sakai2024data} statement by statement.
Recall the prefix sums $L$ and $V$ defined as in \cref{sec:rle}.
For a run interval $I=[i..j]$, write
\(
    \lambda(I):=L[j]-L[i-1]
\)
for its total expanded length and abbreviate
\(
    \widehat S_\alpha(I)
    :=
    \widehat S_\alpha(i,j)
    =V[j]-V[i-1]-\alpha\lambda(I).
\)
Since all run lengths are positive, the prefix-sums of the lengths
\(
    L[0]<L[1]<\cdots<L[z]
\)
are strictly monotonically increasing.

\subsection{Weighted splitting lemma}

We first establish the weighted counterpart of Sakai's Lemma~1~\cite[Lemma~1]{sakai2024data}.
Call a run interval $[i..j]$ \emph{boundary-positive} at offset $\alpha$ if all its nonempty prefixes and suffixes have positive score.
For a nonempty run interval $[i..j]$, let $\widehat\alpha(i,j)$ be the smallest offset at which $[i..j]$ is not boundary-positive, 
and let $\widehat\kappa(i,j)$ be an index $k\in[i..j]$ such that, 
at offset $\widehat\alpha(i,j)$, at least one of the prefix $[i..k]$ and the suffix $[k..j]$ has score zero.

\begin{lemma}[{\cite[Lemma~1]{sakai2024data}}] \label{lem:weighted_sakai_split}
For every offset $\alpha$ and every nonempty run interval $[i..j]$, 
if $[i..j]$ is boundary-positive at offset $\alpha$, 
then $[i..j]$ is its unique maximum-score interval.  
Otherwise, if the maximum score of a nonempty interval is positive, at least one of an arbitrary maximum-score interval of
\( [i..\widehat\kappa(i,j)-1] \)
and
  \(  [\widehat\kappa(i,j)+1..j]\)
is a maximum-score interval of $[i..j]$.
\end{lemma}
\begin{proof}
If $[i..j]$ is boundary-positive, every proper subinterval omits a nonempty prefix, a nonempty suffix, or both, each having positive score.
Hence, the score of every proper subinterval is strictly smaller than that of $[i..j]$.

Now suppose that $[i..j]$ is not boundary-positive at offset $\alpha$.
Set
\( \beta:=\widehat\alpha(i,j)\)
and
\( k:=\widehat\kappa(i,j). \)
Then $\alpha\ge\beta$.
Suppose first that the prefix $[i..k]$ has score zero at offset $\beta$;
the suffix case is symmetric.
By the minimality of $\beta$, every prefix of $[i..j]$ has nonnegative
score at offset $\beta$.
Consequently, for every $g\in[i..k]$,
\(
    \widehat S_\beta(g,k)
    =
    -\widehat S_\beta(i,g-1)
    \le 0.
\)
Since $L[k]-L[g-1]>0$, for every $\alpha\ge\beta$,
\(
    \widehat S_\alpha(g,k)
    =
    \widehat S_\beta(g,k)
    -(\alpha-\beta)\bigl(L[k]-L[g-1]\bigr)
    \le0.
\)
Thus, for every interval $[g..h]$ containing $k$, deleting $[g..k]$ produces the interval $[k+1..h]$, possibly empty, without decreasing its score.  
Hence, we can choose a maximum-score interval of $[i..j]$ that is a subrange of $[i..k-1]$ or $[k+1..j]$.  
Consequently, a maximum-score interval of the corresponding side then has the same score as a maximum-score interval of $[i..j]$.
\end{proof}

\subsection{Weighted split data structure}

Sakai defines the density
\(
    \delta(i,j)
    :=
    \frac{S_0(i,j)}{j-i+1}
\)
and uses it to determine the offset and position at which an interval
first ceases to be boundary-positive.
We replace it by the \emph{weighted density}
\(
    \widehat\delta(i,j)
    :=
    \frac{\widehat S_0(i,j)}{L[j]-L[i-1]}
    =
    \frac{V[j]-V[i-1]}{L[j]-L[i-1]}.
\)
This is the slope of the line joining
\( \bigl(L[i-1],V[i-1]\bigr)\)
with
\( \bigl(L[j],V[j]\bigr). \)

For indices $g\le h$, let $\widehat H(g,h)$ be the lower convex hull of the point set
\(
    \bigl\{\bigl(L[k],V[k]\bigr):k\in[g..h]\bigr\},
\)
and let $\widehat K'(g,h)$ be the sequence, in increasing order, 
of the indices whose points are vertices of $\widehat H(g,h)$.

\begin{lemma}[{\cite[Lemma~2]{sakai2024data}}]
\label{lem:weighted_sakai_tangent}
For all indices $i,g,h$ satisfying $1\le i\le g\le h\le z$, 
a binary search of $\widehat K'(g,h)$ finds an index $k\in[g..h]$ minimizing $\widehat\delta(i,k)$.
\end{lemma}
\begin{proof}
The value $\widehat\delta(i,k)$ is the slope of the line joining
\( \bigl(L[i-1],V[i-1]\bigr)\)
and
\( \bigl(L[k],V[k]\bigr). \)
Because the first coordinates $L[k]$ are strictly increasing, 
the tangent argument in the proof of Sakai's Lemma~2 applies verbatim to $\widehat H(g,h)$, 
and a binary search of its vertices finds a minimizing index.
\end{proof}

The same replacement is made in Sakai's hull-construction algorithm:
the two density comparisons in Lines~8 and~10 of Algorithm~1 use $\widehat\delta$ instead of $\delta$.
The proof of Sakai's Lemma~3 uses only the correctness of these hull tests and the fact that every point is deleted at most once.
It therefore gives an $O(z)$-time construction of the weighted hull forest.

Likewise, the comparisons in Lines~13 and~17--18 of Sakai's Algorithm~2 are performed with $\widehat\delta$.
The path-tracing argument of Sakai's Lemma~4 is otherwise unchanged,
and gives $O(\log^2z)$-time queries for the minimizing weighted prefix density.
Applying the symmetric construction to the reversed weighted array gives the corresponding suffix queries.

\begin{lemma}[{\cite[Theorem~5]{sakai2024data}}] \label{lem:weighted_sakai_split_ds}
The weighted run sequence can be preprocessed in $O(z)$ time and $O(z)$ space such that, 
for every nonempty run interval $[i..j]$, 
the split offset $\widehat\alpha(i,j)$ and a corresponding split position $\widehat\kappa(i,j)$ can be found in $O(\log^2 z)$ time.
\end{lemma}
\begin{proof}
By \cref{lem:weighted_sakai_tangent}, the lower-hull search used for prefixes remains valid after replacing every prefix-sum point 
$\bigl(k,S_0(1,k)\bigr)$ by $\bigl(L[k],V[k]\bigr)$ and 
every density $\delta$ by $\widehat\delta$.
With these replacements, Sakai's Algorithms~1 and~2 and the analyses of Lemmas~3 and~4 are unchanged.
The symmetric construction on the reversed weighted sequence supplies the suffix queries.  
Sakai's combination of the prefix and suffix structures therefore yields the stated split query and the same time and space bounds.
\end{proof}

The split representation of \cref{lem:weighted_sakai_split_ds} computes offsets at which the score of one particular prefix or suffix $I$ becomes zero.  
Since
\(
    \widehat S_\alpha(I)
    =
    \widehat S_0(I)-\alpha\lambda(I),
\)
every such split offset has the form
\(
    \alpha
    =
    \frac{\widehat S_0(I)}{\lambda(I)}.
\)

By contrast, a breakpoint (up to border cases) of the final answer catalog emerges when two distinct candidate intervals $I$ and $J$ have equal scores, 
and therefore has the form
\(
    \alpha
    =
    \frac{\widehat S_0(I)-\widehat S_0(J)}
         {\lambda(I)-\lambda(J)}.
\)
A transition to the empty answer is the special case in which one of the two score functions is identically zero.

\subsection{Candidate intervals}

Sakai's candidate tree $\mathcal T$ is a binary tree defined by recursively splitting a run interval at the pair returned by the split representation of \cref{lem:weighted_sakai_split_ds}.
A \emph{candidate interval} is a nonempty run interval $I=[g..h]$ that labels a vertex of $\mathcal T$. 
Let
\(
    \mathcal C
    :=
    \{
        I :
        I\text{ is a nonempty vertex interval of }\mathcal T
    \}.
\)
For a fixed offset $\alpha$, a non-root candidate interval $I$ with parent $P$ is called \emph{active} if
\(
    \widehat\alpha(P)
    \le
    \alpha
    <
    \widehat\alpha(I).
\)
The root is active when
\(
    \alpha<\widehat\alpha([1..z]).
\)
By the following claim, some active candidate interval is a maximum-score interval:

\begin{lemma}[{\cite[Lemma~6]{sakai2024data}}] \label{lem:weighted_sakai_candidates}
For every offset $\alpha$ with a nonempty answer interval, some active vertex of $\mathcal T$ is a maximum-score interval.
\end{lemma}
\begin{proof}
Use \cref{lem:weighted_sakai_split} at every internal vertex of $\mathcal T$.  
If the interval represented by that vertex is boundary-positive, 
it is itself the unique maximum-score interval.
Otherwise, the lemma guarantees that one of the two recursively considered sides contains a maximum-score interval.  
Following such a side inductively reaches an active vertex representing a maximum-score interval.  
This is the induction used in Sakai's proof of Lemma~6.
\end{proof}

The enumeration in Lines~2--13 of Sakai's Algorithm~3 is unchanged, 
except that it calls the data structure of \cref{lem:weighted_sakai_split_ds}.
Sakai's charging argument associates every internal split generated
during this enumeration with a distinct array index, 
so only $O(z)$ candidate intervals are generated.
We restate Sakai's Algorithm~3 by \cref{alg:weighted-sakai-catalog} for completeness.

\subsection{Grouping candidates}

Sakai's Lemma~7 shows that, among candidate intervals of the same length, 
it suffices to retain one having maximum zero-offset score.
For our weighted problem, the coefficient of the offset is the total expanded length $\lambda(I)$ rather than the length of $I$.
We therefore use the following weighted counterpart.

Let
\(
    \mathcal M
    :=
    \{\lambda(I):I\in\mathcal C\}
\)
be the set of their distinct expanded lengths.
For every $\mu\in\mathcal M$, choose a candidate interval $I_\mu$ such
that
\(
    \lambda(I_\mu)=\mu
\)
and $\widehat S_0(I_\mu)$ is maximum among all candidate intervals of expanded length $\mu$.

\begin{lemma}[{\cite[Lemma~7]{sakai2024data}}] \label{lem:weighted_sakai_grouping}
For every offset $\alpha$ and every candidate interval $J\in\mathcal C$, 
there is an $I_\mu$ with $\widehat S_\alpha(J) \le \widehat S_\alpha(I_\mu)$. 
Consequently, whenever the WPOMSS answer is nonempty, some maximum-score candidate belongs to 
$\{I_\mu : \mu\in M\}$.
\end{lemma}
\begin{proof}
Let $J\in\mathcal C$ and set $\mu:=\lambda(J)$.
By the choice of $I_\mu$,
\(
    \widehat S_0(J)
    \le
    \widehat S_0(I_\mu).
\)
Since $J$ and $I_\mu$ have the same expanded length, their score
functions have the same coefficient of $\alpha$.  Therefore, for every
offset $\alpha$,
\(
    \widehat S_\alpha(J)
    =
    \widehat S_0(J)-\alpha\mu
    \le
    \widehat S_0(I_\mu)-\alpha\mu
    =
    \widehat S_\alpha(I_\mu).
\)
Thus, replacing $J$ by $I_\mu$ never decreases the score, which proves both claims.
\end{proof}

Sakai stores one representative for each possible interval length $|I| \in[1..z]$ in an array indexed by $|I|$.
In our weighted setting, candidates are instead grouped by their expanded length
\( \lambda(I)=\sum_{i=g}^{h}\ell_i, \)
which can be as large as $n$.
Since only $O(z)$ candidate intervals are generated, 
we store representatives only for the expanded lengths that occur, 
using a balanced dictionary keyed by $\lambda(I)$.

\subsection{Final answer catalog}

Let
\(
    \mathcal P
    :=
    \{(0,0)\}
    \cup
    \left\{
        \bigl(\mu,\widehat S_0(I_\mu)\bigr):
        \mu\in\mathcal M
    \right\}.
\)

\begin{lemma}[{\cite[Theorem~8]{sakai2024data}}]
\label{lem:weighted_sakai_catalog}
The upper convex hull of $\mathcal P$, together with the candidate
interval associated with each nonzero hull vertex, represents the
WPOMSS answer catalog.
\end{lemma}
\begin{proof}
For every $\mu\in\mathcal M$ and every offset $\alpha$,
\(
    \widehat S_\alpha(I_\mu)
    =
    \widehat S_0(I_\mu)-\alpha\mu,
\)
which is the value of the linear function
\(
    (x,y) \mapsto y-\alpha x
\)
at the point $\bigl(\mu,\widehat S_0(I_\mu)\bigr)$.
By \cref{lem:weighted_sakai_candidates,lem:weighted_sakai_grouping},
maximizing this function over $\mathcal P$ therefore returns a maximum-score candidate interval, 
or $(0,0)$ when no nonempty candidate has positive score.
The maximizers of a linear function over a finite point set lie on its upper convex hull.

Let $(\mu_1,s_1)$ and $(\mu_2,s_2)$ be consecutive hull vertices with $\mu_1\ne\mu_2$.
Their scores agree precisely when
\(
    s_1-\alpha\mu_1=s_2-\alpha\mu_2.
\)
This happens for $\alpha = \frac{s_2-s_1}{\mu_2-\mu_1}$.

These intersection offsets, sorted increasingly, are the breakpoints of the answer catalog.  
At a breakpoint, assign the breakpoint to the interval active immediately to its right.  
Since increasing $\alpha$ penalizes larger expanded lengths more strongly, 
this interval has the smaller expanded length and is therefore a shortest optimum at the breakpoint.  
If several intervals define the same score function, 
they have the same expanded length and zero-offset score, 
and any one of them may be retained.

The point $(0,0)$ represents the empty interval and its score is identically zero.  
Hence, it is selected exactly when no nonempty candidate has positive score.  
This gives precisely the catalog constructed from the upper hull in Lines~14--22 of Sakai's Algorithm~3.
\end{proof}

We can finally prove the main result of this section.

\thmWOPMSS*
\begin{proof}
There are $O(z)$ candidate intervals.
By \cref{lem:weighted_sakai_split_ds}, their enumeration performs $O(z)$ weighted split queries, 
each taking $O(\log^2 z)$ time.
The dictionary and sorting require $O( z\log z)$ time, 
and the final upper hull is constructed in $O(z)$ time.
Thus, the total construction time is
\( O(z \log^2 z), \)
and the space consumption is $O(z)$.
By \cref{lem:weighted_sakai_catalog}, the resulting answer catalog has $O(z)$ entries, 
so a predecessor search in the catalog answers a WPOMSS query in $O(\log z)$ time.
\end{proof}

\begin{algorithm}[t]
\caption{Construction of the WPOMSS answer catalog \cite[Algorithm~3]{sakai2024data}}
\label{alg:weighted-sakai-catalog}
\begin{algorithmic}[1]
\Require Weighted run sequence
\(
    (\ell_1,x_1),\ldots,(\ell_z,x_z)
\)
\Ensure WPOMSS answer catalog

\State Construct the weighted split data structure of
       \cref{lem:weighted_sakai_split_ds}

\State $\mathcal C\gets$ the candidate intervals obtained by traversing $\mathcal T$ as in \cref{lem:weighted_sakai_candidates}

\State $\mathsf{rep}\gets$ an empty balanced dictionary, \textsf{keys}: expanded lengths, \textsf{values}: candidate intervals

\ForAll{$I\in\mathcal C$}
    \State $\mu\gets\lambda(I)$
    \If{$\mu\notin \mathsf{rep}.\mathsf{keys}$ \textbf{or}
        $\widehat S_0(I)>
         \widehat S_0(\mathsf{rep}[\mu])$}
        \State $\mathsf{rep}[\mu]\gets I$
    \EndIf
\EndFor

\State
\(
    \mathcal P
    \gets
    \{((0,0),\emptyset)\}
    \cup
    \left\{
        \left(
            \bigl(\mu,
                  \widehat S_0(\mathsf{rep}[\mu])\bigr),
            \mathsf{rep}[\mu]
        \right)
        :
        \mu\in\operatorname{dom}(\mathsf{rep})
    \right\}
\)
\State $\mathcal H\gets\Call{UpperHull}{\mathcal P}$
       \Comment{Hull vertices retain their interval labels}

\State \Return $\Call{RightActiveCatalog}{\mathcal H}$
\end{algorithmic}
\end{algorithm}

\end{document}